%% file: polytime-part-oracles.tex
\documentclass[11pt]{article}
\usepackage[english]{babel}
\usepackage{graphicx}
\usepackage{caption}
\usepackage{subcaption}
\usepackage{scrextend}
\usepackage{thmtools}
\usepackage[normalem]{ulem}
\newcommand{\ds}{\displaystyle}

\newcommand{\minp}{\rho}
\newcommand{\len}{\ell}
\newcommand{\sizefrac}{\beta}

\newcommand{\minpexp}{3000}
\newcommand{\delexp}{3100}
\newcommand{\alphaexp}{4/3}
\newcommand{\lenexp}{30}
\newcommand{\phiexp}{10}

\newcommand{\minpval}{d^{-60}\eps^{\minpexp}}
\newcommand{\lenval}{d^6 \eps^{-\lenexp}}
\newcommand{\sizefracval}{\eps/10}
\newcommand{\deltaval}{d^{-70}\eps^{\delexp}}
\newcommand{\alphaval}{\frac{\eps^{\alphaexp}}{300,000}}
\newcommand{\phival}{\eps^{\phiexp}}

\newcommand{\trwalk}{\widehat{M}}
\newcommand{\trunp}[3]{\widehat{p}_{#1, #2}(#3)}

\newcommand{\pvector}[3]{p_{#1, #2}(#3)}
\newcommand{\ball}{IB}
\newcommand{\cluster}{\texttt{cluster}}
\newcommand{\lev}[2]{L_{#1, #2}}
\newcommand{\level}[3]{L_{#1, #2, #3}}
\newcommand{\free}[1]{F_{#1}}
\newcommand{\prw}[1]{\widehat{p}_{#1}}
\newcommand{\numitr}{100 \log {\len}}

\newcommand{\maxphase}{2 \delta^{-1}\log(\delta^{-1})}
\newcommand{\bucket}[2]{B_{#1, #2}}

\newcommand{\ph}[1]{h_#1}
\newcommand{\ST}{K}
\newcommand{\sizethresh}{k}

\newcommand{\ls}[2]{I_{#1, #2}}
\newcommand{\lin}[3]{f_{#1, #2, #3}}

\newcommand{\isfree}{{\tt IsFree}}
\newcommand{\findr}{{\tt findr}}
\newcommand{\findanchor}{{\tt findAnchor}}
\newcommand{\findpart}{{\tt findPartition}}
\newcommand{\findball}{{\tt findIB}}
\newcommand{\globpart}{{\tt globalPartition}}

\newcommand{\clustercode}{
\noindent
\fbox{\begin{minipage}{0.9\textwidth}
    {\cluster$(v, t, k)$}
    \begin{compactenum}
    \item Determine $\trwalk^{t}\vec{v}$
        \item For all $k' \in [k, 2k]$ calculate $\Phi(\level{v}{t}{k'})$.
        \item Find the largest $k' \in [k,2k]$ (if any) with the following properties: $\Phi(\level{v}{t}{k'} \cup \{v\}) \leq \phi$ and $\level{v}{t}{k'} \in \supp(\trwalk^{t}\vec{v})$.
        \item If such a $k'$ exists, set $C := \level{v}{t}{k'} \cup \{v\}$, else $C := \{v\}$.
        \item Return $C$.
    \end{compactenum}
\end{minipage}
}
}

\newcommand{\parameterinfo}{
\begin{itemize}
    \item $\minp = \minpval$: Minimum probability for truncation.
    \item $\len = \lenval$: Maximum random walk length.
    \item $\sizefrac = \sizefracval$: Unclustered fraction cutoff.
    \item $\delta = \deltaval$: Phase probability. 
    \item $\alpha = \alphaval$: Heavy bucket parameter.
    \item $\phi = \phival$: Conductance parameter.
\end{itemize}
}

\input{preamble.tex}

\title{Random walks and forbidden minors III: $\poly(d\eps^{-1})$-time partition oracles for minor-free graph classes}
\author{Akash Kumar\thanks{Department of Computer Science, EPFL. {\href{mailto:akash.kumar@epfl.ch}{akash.kumar@epfl.ch}} 
This project has received funding from the European Research Council (ERC) under the European Union’s Horizon 2020 research and innovation programme (grant agreement No 759471).}
\and C. Seshadhri\thanks{Department of Computer Science, University of California, Santa Cruz. {\href{mailto:sesh@ucsc.edu}{sesh@ucsc.edu}}}
\and Andrew Stolman\thanks{Department of Computer Science, University of California, Santa Cruz. { \href{mailto:astolman@ucsc.edu}{astolman@ucsc.edu}}\newline {CS and AS acknowledge the support of NSF grants CCF-1740850, CCF-1813165, CCF-1909790, CCF-2023495, and ARO Award W911NF1910294.}} }

\date{}
\begin{document}

\maketitle
\thispagestyle{empty}

\abstract{Consider the family of bounded degree graphs in any minor-closed family (such as planar graphs).
Let $d$ be the degree bound and $n$ be the number of vertices of such a graph.
Graphs in these classes have hyperfinite decompositions,
where, for a sufficiently small $\eps > 0$, one removes $\eps dn$ edges to get connected
components of size independent of $n$. An important tool for sublinear algorithms and property testing for such classes is
the \emph{partition oracle}, introduced by the seminal work of Hassidim-Kelner-Nguyen-Onak (FOCS 2009).
A partition oracle is a local procedure that gives
consistent access to a hyperfinite decomposition, without any preprocessing. Given a query vertex $v$, the 
partition oracle outputs the component containing $v$ in time independent of $n$. All the answers
are consistent with a single hyperfinite decomposition.

The partition oracle of Hassidim et al. runs in time $d^{\poly(d\eps^{-1})}$ per query. They
pose the open problem of whether $\poly(d\eps^{-1})$-time partition oracles exist. Levi-Ron (ICALP 2013)
give a refinement of the previous approach, to get a partition oracle that runs in time $d^{\log(d\eps^{-1})}$-per query.

In this paper, we resolve this open problem and give $\poly(d\eps^{-1})$-time partition oracles
for bounded degree graphs in any minor-closed family. Unlike the previous line of work based on 
combinatorial methods, we employ techniques from spectral graph theory. We build on a recent spectral
graph theoretical toolkit for minor-closed graph families, introduced by the authors to develop efficient property testers.
A consequence of our result is a $\poly(d\eps^{-1})$-query tester for \emph{any monotone and additive} 
property of minor-closed families (such as bipartite planar graphs). Our result also gives $\poly(d\eps^{-1})$-query
algorithms for additive $\eps n$-approximations for problems such as maximum matching, minimum vertex cover, maximum independent set,
and minimum dominating set for these graph families.

}

\input{intro.tex}

\section{Global partitioning and its local implementation} \label{sec:prelims}


There are a number of parameters that are used in the algorithm. We list them out here
for reference. It is convenient to fix the value of $\eps$ in advance, so that all
the values of the following parameters are fixed. Note that all these parameters are polynomial
is $\eps$. We will express all running times as polynomials in these parameters, ensuring all running time are $\poly(\eps^{-1})$.

\parameterinfo

\subsection{Truncated diffusion} \label{sec:trunc}

The main process used to find sets of the partition is a \emph{truncated diffusion}.
We assume that the input graph $G$ is connected, has $n$ vertices, and degree bound $d$.
Define the symmetric random walk matrix $M$ as follows. For every edge $(u,v)$,
$M_{u,v} = M_{v,u} = 1/2d$. For every vertex $v$, $M_{v,v} = 1-d(v)/2d$, where $d(v)$
is the degree of $v$. The matrix $M$ is doubly stochastic, symmetric, and the (unique) stationary
distribution is the uniform distribution.

Given a vector $\vec{x} \in (\RR^+)^n$, diffusion is the evolution $M^t\vec{x}$. We define
a truncated version, where after every step, small values are removed.
For any vector $\vec{x}$, let $\supp(\vec{x})$
denote the support of the vector.

\begin{definition} \label{def:trun}	Define the operator $\trwalk \colon (\R^+)^n \to (\R^+)^n$ as follows.
For $\vec{x} \in (\R^+)^n$,
the vector $\trwalk\vec{x}$ is obtained by zeroing out all coordinates in $M\vec{x}$ whose value is at most $\minp$. 

For $t > 1$, the operator $\trwalk^t$ is the $t$-step truncated diffusion, and is recursively defined as $\trwalk(\trwalk^{t-1}\vec{x})$.

Define $\trunp{v}{t}{w}$ to be the coordinate corresponding to vertex $w$ in the $t$-step truncated diffusion
starting from vertex $v$.
\end{definition}

We stress that the $t$-step truncated diffusion is obtained from a standard diffusion by truncating low
values at \emph{every} step of the diffusion. Note that as the truncated diffusion progresses,
the $l_1$-norm of the vector may decrease at each step. Importantly, for any distribution vector $\vec{x}$,
$\supp(\trwalk^t \vec{x})$ has size at most $\minp^{-1}$. We heavily use this property in our running time 
analysis.

We define \emph{level sets}, a standard concept in spectral partitioning algorithms.
Somewhat abusing notation, for vertex $v \in V$, we use $\vec{v}$ to denote the unit
vector in $(\R^+)^n$ corresponding to the vertex $v$. (We never use to vector notation
for any other kind of vectors.)

\begin{definition} \label{sec:level} For vertex $v \in V$, length $t$, and threshold $k$,
let $\level{v}{t}{k}$ be the set of vertices corresponding to the $k$ largest coordinates
in $\trwalk^t\vec{v}$ (ties are broken by vertex id).

For any set $S$ of vertices, the conductance of $S$ is $\Phi(S) := E(S,\overline{S})/[2\min(|S|,|\overline{S}|)d]$.
(We use $E(S,\overline{S})$ to denote the number of edges between $S$ and its complement.)
\end{definition}

We describe the key subroutine that finds low conductance cuts. It performs a sweep cut over the truncated diffusion vector.

\medskip
\noindent \clustercode \\

\begin{claim} \label{clm:cluster} The procedure \cluster$(v,t,k)$ runs in time $O(\minp^{-1}td\log(\minp^{-1}td) + kd\log k)$. 
The output set $C$ has the following properties. (i) $v \in C$. (ii) 
If $C$ is not a singleton, then $|C| \in [k,2k]$, $\Phi(C) \leq \phi$, and $C \subseteq \supp(\trwalk^t\vec{v})$.
\end{claim}

\begin{proof} The latter properties are apparent from the description of \cluster.

We analyze the running time. The convenience of the truncated diffusion
is that it can computed exactly by a deterministic process. First, for any $b \geq 1$, we show
that the running time to compute $\trwalk^b\vec{v}$ is $O(\minp^{-1}bd)$. 
Note that for any $t$, $\supp(\trwalk^b\vec{v})$ has size at most $\minp^{-1}$,
since $M$ is a stochastic matrix and all non-zero entries in $\trwalk^b\vec{v}$
have value at least $\minp$. Given the vector $\trwalk^b\vec{v}$, the vector $\trwalk^{b+1}\vec{v}$
can be computed by determining $M\trwalk^b\vec{v}$ and then zeroing out coordinates that are less than $\minp$.
This process can be done in $O(d| \supp(\trwalk^b\vec{v})|) = O(\minp^{-1}d)$.
By summing this running time over all timesteps, we get that the total time is $O(\minp^{-1}bd)$.

Thus, $\trwalk^t\vec{v}$ can be computed exactly in $O(\minp^{-1}td)$ time. To compute the level sets,
one can sort the coordinates of this vector (breaking ties by id), and process them in decreasing order. One can
iteratively store $\level{v}{t}{k}$ in a dictionary data structure. Given $\Phi(\level{v}{t}{k})$,
one can compute $\Phi(\level{v}{t}{k+1})$ by $O(d)$ lookups into the dictionary. The total
running time of this step is $O(kd\log k)$.
\end{proof}

\subsection{The global partitioning procedure} \label{sec:global}

The global partitioning procedure \globpart{} will output a partition of the vertices
satisfying the conditions in \Def{oracle}. This global procedure will run in linear time.
In the next subsection, we show
how the output of the global procedure can be generated locally in $\poly(\eps^{-1})$ time, thereby giving
us the desired partition oracle. It will be significantly easier to understand and analyze the partition
properties of the global procedure.

The key ingredient in \globpart{} that allows for a local implementation is a \emph{preprocessing}
step. The preprocessing allows for the ``coordination" required for consistency of various
local partitioning steps. All the randomness is used in the preprocessing, after which
the actual partitioning is deterministic. The job of the preprocessing 
is to find the following sets of values, which are used for two goals: (i) ordering
vertices, (ii) setting parameters for calls to \cluster.

The preprocessing generates, for all vertices $v$, the following values.
\begin{asparaitem}
    \item $h_v$: The \emph{phase} of $v$.
    \item $k_v$: The size threshold of $v$.
    \item $t_v$: The walk length of $v$.
\end{asparaitem} 

\medskip
Before giving the procedure description, we explain how these values are generated.

{\em Phases:} For each $v$, $h_v$ is set to $\max(X,\overline{h})$,
where $X$ is independently sampled from $Geo(\delta)$, the geometric distribution with parameter $\delta$. 
Moreover $\overline{h} := \maxphase$, so the maximum phase value is capped. 

{\em Size thresholds:} The computation of these thresholds is the most complex part of our algorithm
(and analysis), and is the ``magic ingredient" that makes the partition oracle possible. 
We first run a procedure $\findr$ that runs in $\poly(\eps^{-1})$ time and outputs
a set of \emph{phase size thresholds} $k_1, k_2, \ldots, k_{\overline{h}}$. All the thresholds
have value at most $\minp^{-1}$ and $k_{\overline{h}}$ will be zero.
The (involved) description
of \findr{} and its properties are in \Sec{findr}. For now, it suffices to say
that its running time is $\poly(\eps^{-1})$, and that it outputs phase size thresholds. The
size threshold for a vertex $v$ is simply $k_{h_v}$, corresponding to the phase it belongs to.

{\em Walk lengths:} These are simply chosen independently and uniformly in $[1,\ell]$.

\medskip
The analysis is more transparent when we assume that all the randomness used by the algorithm
is in a random seed $\bR$, of $O(n\cdot\poly(\eps^{-1}))$ length. The seed $\bR$ is passed as an argument
to the partitioning procedure, which uses $\bR$ to generate all the values described above.
(For convenience, we will assume random access to the adjacency list of $G$, without passing
the graph as a parameter.)

It is convenient to define an ordering on the vertices, given these values. For cleaner
notation, we drop the dependence on $\bR$.

\begin{definition} \label{def:order} For vertex $u, v \in V$, we say that $u \prec v$
if: $h_u < h_v$ or if $h_u = h_v$, the id of $v$ is less than that of $v$.
\end{definition}

\fbox{
    \begin{minipage}{0.9\textwidth}
    {\globpart$(\bR)$} \\
    \smallskip
    Preprocessing:
    \begin{compactenum}
       \item For every $v \in V$:
        \begin{compactenum}
            \item Use $\bR$ to set $h_v := \max(X,\overline{h})$ ($X \sim Geo(\delta)$).
            \item Use $\bR$ to set $t_v$ uniform random in $[1,\ell]$.
        \end{compactenum}
        \item Call \findr$(\bR)$ to generate values $k_1, k_2, \ldots, k_{\overline{h}}$. For every $v \in V$, set $\sizethresh_v = \sizethresh_{h_v}$.
    \end{compactenum}
    \smallskip
    Partitioning:
    \begin{compactenum}
        \item Initialize the partition $\bP$ as an empty collection. Initialize the free set $F := V$.
        \item For all vertices in $V$ in increasing order of $\prec$:
        \begin{compactenum}
            \item Compute $C = \cluster(v, t_v, k_v)$.
            \item Add the connected components of $C \cap F$ to the partition $\bP$.
            \item Reset $F = F \setminus C$.
        \end{compactenum}
        \item Output $\bP$.
    \end{compactenum}
    \end{minipage}
}

\medskip

Since all of our subsequent discussions are about \globpart, we abuse notation
assuming that the preprocessing is fixed. We refer to \cluster$(v)$ to denote \cluster$(v, t_v, k_v)$.
These are the only calls to cluster that are ever discussed, so it is convenient to just parametrize by
the vertex argument. Furthermore, for ease of notation, we sometimes refer to the output of the procedure as \cluster$(v)$.

We observe that the output $\bP$ is indeed a partition of $V$ into connected components. At any intermediate
step, the free set $F$ is precisely the set of vertices that have not been assigned to a cluster.
Note that \cluster$(v)$ always contains $v$ (\Clm{cluster}), so all vertices eventually enter (the sets of) $\bP$.

We note that $v$ might not be in $F$ when \cluster$(v)$ is called. This may lead
to new components in $\bP$ that do not involve $v$, which may actually not be low
conductance cuts. This may seem like an oversight: why initiate diffusion clusters
from vertices that are already partitioned? Many challenges in our analysis arise from such
clusters. On the other hand, such an ``oblivious" partitioning scheme leads to a simple local implementation. 

\subsection{The local implementation} \label{sec:local}

A useful definition in the local implementation is that of \emph{anchors} of vertices.
As mentioned earlier, we fix the output of the preprocessing (which is equivalent
to fixing $\bR$). 

\begin{definition} \label{def:anchor} Consider the running of \globpart$(\bR)$. 
The \emph{anchor} of $w$ is the (unique) vertex $w$ such that the component in $\bP$
containing $v$ was created by the call to \cluster$(v)$.
\end{definition}

Suppose we label every vertex by its anchor. We can easily determine
the sets of $\bP$ locally.

\begin{claim} \label{clm:anchor} The sets of $\bP$ are exactly the maximal
connected components of vertices with the same anchor.
\end{claim}

\begin{proof} We prove by induction over the $\prec$ ordering of vertices.
The base case is vacuously true. Suppose, just before $v$ is considered,
all current sets in $\bP$ are maximal connected components with the same anchor, which cannot be $v$. 
No vertex in $F$ can have an anchor yet; otherwise, it would be clustered and part of (a set in) $\bP$.
All the new vertices clustered have $v$ as anchor. Moreover, the sets added to $\bP$ are precisely
the maximal connected components with $v$ as anchor. 
\end{proof}

We come to a critical definition that allows for searching for anchors.
We define the ``inverse ball" of a vertex: this is the set of all vertices that reach
$v$ through truncated diffusions. We note that reachability is not symmetric, because
the diffusion is truncated at every step.

\begin{definition} \label{def:ball}
	For $v \in V$, let $\ball(v) = \{w \ | \ \exists t \in [0,\len], v \in \supp(\trwalk^t\vec{w})\}$.
\end{definition}

\begin{claim} \label{clm:ballsize} $|\ball(v)| \leq \len\minp^{-1}$.
\end{claim}

\begin{proof} All vertices $w \in \ball(v)$ have the property
that (for some $t \leq \len$) $\trunp{w}{t}{v} \neq 0$. That implies that $\pvector{w}{t}{v} \geq \minp$.
By the symmetry of the random walk, $\pvector{v}{t}{w} \geq \minp$. For any fixed $t$, there are at most $\minp^{-1}$
such vertices $w$. Overall, there can be at most $\len\minp^{-1}$ vertices in $\ball(v)$.
\end{proof}

Now we have a simple characterization of the anchor that allows for local implementations.

\begin{lemma} \label{lem:anchor} The anchor of $v$ is the smallest vertex (according to $\prec$)
in the set $\{s | s \in \ball(v) \ \textrm{and} \ v \in \textrm{\cluster}(s)\}$.
\end{lemma}

\begin{proof} Let the anchor of $v$ be the vertex $u$. We first argue that $u$ in the given set.
Clearly, $v \in \cluster(u)$.  If $u=v$, then $u = v \in \ball(v)$ and we are done. Suppsoe $u \neq v$.
Then $\cluster(u)$ is not a singleton (since it contains $v$). By \Clm{cluster}, $\cluster(u)$
is contained in the support of $\trwalk^{t_v}\vec{u}$, implying that $v \in \supp(\trwalk^{t_v}\vec{u})$.
Thus, $u \in \ball(v)$ and the anchor $u$ is present in the given set.

It remains to argue that $u$ is the smallest such vertex. Suppose there exists $u' \prec u$ in this set.
In \globpart{}, \cluster$(u')$ is called before \cluster$(u)$. At the end of this call, $v$ is partitioned
and would have $u'$ as its anchor. Contradiction.
\end{proof}

We are set for the local implementation. For a vertex $v$, we compute $\ball(v)$ and
run \cluster$(u)$ for all $u \in \ball(v)$. By \Lem{anchor}, we can compute the anchor of $v$,
and by \Clm{anchor}, we can perform a BFS to find all connected vertices with the same anchor. 

We begin by a procedure that computes $\ball(v)$. Since the truncated diffusion is not symmetric,
this requires a little care. We use $N(u)$ to denote the neighborhood of vertex $u$.

\medskip

\fbox{
\begin{minipage}{0.9\textwidth}
    {\findball$(v)$}
    \smallskip
    \begin{compactenum}
        \item Initialize $S = \{v\}$.
        \item For every $t = 1, \ldots, \ell$:
        \begin{compactenum}
            \item For every $w \in S \cup N(S)$, compute $\trwalk^t\vec{w}$. If $v \in \supp(\trwalk^t\vec{w})$, add $v$ to $S$.
        \end{compactenum}
        \item Return $S$.
    \end{compactenum}
\end{minipage}
}

\medskip

\begin{claim} \label{clm:findball} The output of \findball$(v)$ is $\ball(v)$. The running time
is $O(d^2\len^3\minp^{-2})$.
\end{claim}

\begin{proof} We prove by induction on $t$, that after $t$ iterations of the loop,
$S$ is the set $\{w \ | \ \exists t' \in [0,t], v \in \supp(\trwalk^{t'}\vec{w})\}$.
The base case $t=0$ holds because $S$ is initialized to $\{v\}$. Now for the induction.
Consider some $w$ such that $v \in \supp(\trwalk^{t+1}\vec{w})$. 
This means that $(1-d(w)/2d)\trunp{w}{t}{v} + (1/2d)\sum_{w' \in N(w)} \trunp{w'}{t}{v} \geq \minp$.
Since the LHS is an average, for some $w' \in N(w) \cap \{w\}$, $\trunp{w'}{t}{v} \geq \minp$.
Hence, $v \in \supp(\trwalk^t\vec{w'})$, and by induction $w' \in S$ at the beginning of the $(t+1)$th iteration.
The inner loop will consider $w$ (as it is either $w'$ or a neighbor of $w'$), correctly determine that $v \in \supp(\trwalk^{t+1}\vec{w})$,
and add it to $S$. By construction, every (new) vertex $w$ added to $S$ has the property that $v \in \supp(\trwalk^{t+1}\vec{w})$.
This completes the induction and the output property.

For the running time, observe that for all iterations, $S \subseteq \ball(v)$. By \Clm{ballsize},
$|S| \leq \len\minp^{-1}$. Hence, $|S \cup N(S)|$ has size $O(d\len\minp^{-1})$.
The computation of each $\trwalk^t\vec{w}$ can be done in $O(d\len\minp^{-1})$ time, since the
distribution vector after each step has support size at most $\minp^{-1}$. The total
running time of each iteration is $O(d^2\len^2\minp^{-2})$. There are at most $\len$ iterations,
leading to a total running time of $O(d^2\len^3\minp^{-2})$.

\end{proof}

We can now describe the local partitioning oracle (modulo the description of \findr).

\medskip

\fbox{
\begin{minipage}{0.9\textwidth}
    {\findanchor$(v, \bR)$}
    \smallskip
    \begin{compactenum}
        \item Run \findr$(\bR)$ to get the set $K = \{k_1, k_2, \ldots, k_{\overline{h}}\}$.
        \item Run \findball$(v)$ to compute $\ball(v)$.
        \item Initialize $A = \emptyset$.
        \item For every $s \in \ball(v)$:
        \begin{compactenum}
            \item Using $\bR$ determine $h_s, t_s$. Using $K$, determine $k_s$.
            \item Compute $C = \cluster(s,t_s,k_s)$.
            \item If $C \ni v$, then add $s$ to $A$.
        \end{compactenum}
        \item Output the smallest vertex according to $\prec$ in $A$.
    \end{compactenum}
\end{minipage}
}

\medskip

\fbox{
\begin{minipage}{0.9\textwidth}
    {\findpart$(v, \bR)$}
    \smallskip
    \begin{compactenum}
        \item Call \findanchor$(v,\bR)$ to get the anchor $s$.
        \item Perform BFS from $v$. For every vertex $w$ encountered, first call \findanchor$(w,\bR)$.
        If the anchor is $s$, add $w$ to the BFS queue (else, ignore $w$).
        \item Output the set of vertices that entered the BFS queue.
    \end{compactenum}
\end{minipage}
} \\

The following claim is a direct consequence of \Lem{anchor} and \Clm{findball}.

\begin{claim} \label{clm:findanchor} The procedure \findanchor$(v,\bR)$ outputs
the anchor of $v$ and runs in time $O((d\len\minp^{-1})^3)$ plus the running time of \findr.
\end{claim}

\begin{proof} Observe that \findanchor$(v,\bR)$ finds $\ball(v)$,
computes \cluster$(s)$ for each $s \in \ball(v)$, and outputs the smallest (by $\prec$) $s$
such that $v \in \cluster(s)$. By \Lem{anchor}, the output is the anchor of $v$.

By \Clm{findball}, the running time of \findball$(v)$ is $O(d^2\len^3\minp^{-2})$. The number of 
calls to \cluster{} is $|\ball(v)|$, which is at most $\len\minp^{-1}$ (\Clm{ballsize}). Each call to \cluster{} runs
in time $O(d\len\minp^{-2})$, by \Clm{cluster} and the fact that $k_s \leq \minp^{-1})$.
Ignoring the call to \findr, the total running time is $O(d^2\len^3\minp^{-3})$.
\end{proof}

\begin{theorem} \label{thm:findpart} The output of \findpart$(v,\bR)$ is precisely the set in $\bP$
containing $v$, where $\bP$ is the partition output by \globpart$(\bR)$. The running time
of \findpart$(v,\bR)$ is $O((d\len\minp^{-1})^4)$ plus the running time of \findr.
\end{theorem}

\begin{proof} By \Clm{findanchor}, \findanchor{} correctly outputs the anchor. By \Clm{anchor},
the set $S$ in $\bP$ containing $v$ is exactly the maximal connected component of vertices sharing
the same anchor (as $v$). The set $S$ in $\bP$ is generated in \globpart$(\bR)$ by a call
to \cluster, whose output is a set of size at most $\minp^{-1}$. 
The total number of calls to \findanchor{} made by \findpart$(v,\bR)$ is at most $d\minp^{-1}$,
since a call is made to either a vertex in the set $S$ or a neighbor of $S$.
Overall, the total running time is $O((d\len\minp^{-1})^5)$ plus the running time of \findr. (Instead
of calling \findr{} in each call to \findanchor, one can simply store its output.)
\end{proof}

\section{Coordination through the size thresholds: the procedure \findr} \label{sec:findr}

We now come to the heart of our algorithm; coordination through \findr. This section gives the crucial ingredient in arguing that the partitioning
scheme does not cut too many edges. The
ordering of vertices (to form clusters) is chosen independent of the graph structure. It is highly likely that,
as the partitioning proceeds, newer \cluster$(v)$ sets overlap heavily with the existing partition. Such clusters
may cut many new edges, without clustering enough vertices. Note that \cluster$(v)$ is a low conductance cut
only in the original graph; it might have high conductance restricted to $F$ (the current free set).

To deal with such ``bad" clusters, we need to prove that every so often, \cluster$(v)$ will successfully partition
enough new vertices. Such ``good" clusters allow the partitioning scheme to suffer many bad clusters.
This argument is finally carried about by a careful charging argument. First, we need to argue that such
good clusters exist. The key tool is given by the following theorem, which is proved using
spectral graph theoretic methods. We state the theorem as an independent statement.

\begin{theorem} \label{thm:restrict-cut} Let $G$ be a bounded degree graph in a minor-closed family.
Let $F$ be an arbitrary set of vertices of size at least $\sizefrac n$.
There exists a size threshold $k \leq \minp^{-1}$ such that the following holds. 
For at least $(\sizefrac^2/\log^2\sizefrac^{-1})n$
vertices $s \in F$, there are at least $(\sizefrac/\log^2\sizefrac^{-1})\len$ timesteps $t \leq \len$ such that:
there exists $k' \in [k,2k]$ such that (i) $\level{s}{t}{k'} \subseteq \supp(\widehat{M}^t\vec{s})$, (ii) $\Phi(\level{s}{t}{k'} \cup \{s\}) < \phi$, and (iii) $|\level{s}{t}{k'} \cap F| \geq \sizefrac^3 k$.
\end{theorem}

The proof of this theorem is deferred to \Sec{ls-cluster}. In this section, we apply this theorem to complete
the description of the partition oracle and prove its guarantees.

We discuss the significance of this theorem. The diffusion used to define $\level{s}{t}{k'}$ occurs in $G$, but
we are promised a low conductance cut with non-trivial intersection with $F$ (since $\phi \ll \sizefrac^3$).
Moreover, such cuts are obtained for a non-trivial
fraction of timesteps, so we can choice one uar. Given oracle access to membership in $F$, it is fairly easy
to find such a size threshold by random sampling.

{\em The importance of phases:} Recall the global partitioning procedure \globpart. We can think of the partitioning process
as divided into phases, where the $h$th phase involves calling \cluster$(v,t_v,\sizethresh_v)$ for all vertices $v$
whose phase value is $h$. Consider the free set at the beginning of a phase $h$, denoting it $F_h$. We apply \Thm{restrict-cut}
to determine the size threshold $\sizethresh_h$. Since all $\sizethresh_v$ values in this phases are precisely $\sizethresh_h$,
this size threshold ``coordinates" all clusters in this phase. As the phase proceeds, the free set shrinks, and the size
threshold $\sizethresh_h$ stops satisfying the properties of \Thm{restrict-cut}. Roughly speaking, at this point, we start
a new phase $h+1$, and recompute the size threshold. The frequency of recomputation is chosen carefully to ensure
that the total running time remains $\poly(\eps^{-1})$.

We now discuss the randomness involved in selecting phases and why geometric random variables are used. Recall
that $h_v$ is independently (for all $v$) set to be $\min(X, \overline{h})$, where $X \sim Geo(\delta)$.
We first introduce some notation regarding phases.

\begin{definition} \label{def:phase-v} The \emph{phase $h$ seeds}, denoted $V_h$, are the vertices whose phase value is $h$.
Formally, $V_h = \{v \ | \ h_v = h\}$. We use $V_{< h}$ to denote $\bigcup_{h' < h} V_h$. (We analogously define $V_{\leq h}, V_{\geq h}$.)

The \emph{free set at phase $h$}, denoted $F_h$, is the free set $F$ in \globpart, just before the first phase $h$
vertex is processed. Formally, $F_h = V \setminus \bigcup_{v \in V_{<h}} \cluster(v)$.
\end{definition}

One can think of the $V_h$s being generated iteratively. Assume that we have fixed the vertices in $V_1, \ldots, V_{h-1}$.
All other vertices are in $V_{\geq h}$, implying that $h_v \geq h$ for such vertices. By the properties of the geometric
random variables, $\Pr[h_v = h+1 | h_v > h] = \delta$. Thus, we can imagine that $V_{h+1}$
is generated by independently sampling each element in $V_{\geq h}$ with $\delta$ probability. We 
restate this observation as \Clm{geo-phase}. \Clm{vh-size} is a simple Chernoff bound argument.

Before proceeding, we state some standard Chernoff bounds (Theorem 1.1 of~\cite{DuPa-book}).

\begin{theorem} \label{thm:chernoff} Let $X_1, X_2, \ldots, X_r$ be independent variables in $[0,1]$. Let $\mu := \EX[\sum_i X_i]$.
\begin{asparaitem}
    \item $\Pr[X \geq 3\mu/2] \leq \exp(-\mu/12)$.
    \item $\Pr[X \leq \mu/2] \leq \exp(-\mu/8)$.
    \item For $t \geq 6\mu$, $\Pr[X \geq t] \leq 2^{-t}$.
\end{asparaitem}
\end{theorem}

\begin{claim} \label{clm:geo-phase} For all $v \in V$ and $1 < h < \overline{h}$, $\Pr[v \in V_h \ | \ v \in V_{\geq h}] = \delta$.
\end{claim}

\begin{claim} \label{clm:vh-size} Let $h < \overline{h}$. Condition on the randomness used to specify $V_1, V_2, \ldots, V_{h-1}$.
Let $S$ be an arbitrary subset of $V_{\geq h}$. 
With probability at least $1-2\exp(-\delta |S|/12)$ over the choice of $V_h$, $|S \cap V_h| \in [\delta|S|/2, 2\delta|S|]$.
\end{claim}

\begin{proof} For every $s \in S$, let $X_s$
be the indicator random variable for $s \in V_h$. By \Clm{geo-phase} and independent phase choices
for each vertex, the $X_s$ are independent Bernoullis with $\delta$ probability. By the Chernoff lower tail of \Thm{chernoff},
$\Pr[\sum_{s \in S} X_s \leq \delta|S|/2] \leq \exp(-\delta|S|/8)$ and $\Pr[\sum_{s \in S} X_s \geq 2\delta|S|] \leq \exp(\delta|S|/12)$.
A union bound completes the proof.
\end{proof}

\begin{claim} \label{clm:last-phase} With probability at least $1-2^{-\delta n}$, $|V_{\overline{h}}| \leq \delta n$.
\end{claim}

\begin{proof} Recall that $\overline{h}$ is the last phase and $\overline{h} = \maxphase$.
The probability that $X \sim Geo(\delta)$ is at least $\maxphase$ is $(1-\delta)^{\maxphase-1} < \delta/6$.
Hence, the probability that any vertex lies in $V_{\overline{h}}$ is at most $\delta/6$ and the expectation
of $V_{\overline{h}}$ is at most $\delta n/6$. . By the Chernoff bound
of \Thm{chernoff}, $\Pr[|V_{\overline{h}}| \geq \delta n] \leq 2^{-\delta n}$.
\end{proof}

With this preamble, we proceed to the description of \findr{} and the main properties of its output.

\subsection{The procedure \findr} \label{sec:sub-findr}

It is convenient to assume that for all $v$, $h_v$ and $t_v$ have been chosen. These quantities are chosen
independently for each vertex using simple distributions, so we will not carry as arguments the randomness used to decide
these quantities.  Recall that the output of \findr{} is the set of size thresholds $\{\sizethresh_1, \sizethresh_2, \ldots, \sizethresh_{\overline{h}}\}$.
It is convenient to use $\ST_h$ to denote $\{\sizethresh_1, \sizethresh_2, \ldots, \sizethresh_{h}\}$. Before describing
\findr, we define a procedure that is a membership oracle for $F_h$.

\medskip

\noindent
\fbox{
\begin{minipage}{0.9\textwidth}
    {\isfree$(u,h,\ST_{h-1})$}
    \smallskip
    \begin{compactenum}
        \item If $h = 1$, output YES.
        \item Run \findball$(u)$ to determine $\ball(u)$. Let $C$ be $\ball(u) \cap V_{< h}$.
        \item Using $\ST_{h-1}$, determine $\sizethresh_v$ for all $v \in C$.
        \item For all $v \in C$, compute \cluster$(v,t_v,\sizethresh_v)$. If the union
        contains $u$, output NO. Else, output YES.\label{step:union} 
    \end{compactenum}
\end{minipage}
}
\medskip

\begin{claim} \label{clm:isfree} Assume that $\ST_{h-1}$ is provided correctly. Then
\isfree$(v,h,\ST_{h-1})$ outputs YES iff $v \in F_h$. The running time is $O((d\len\minp^{-1})^3)$.
\end{claim}

\begin{proof} If $h=1$, then all vertices are free (this is the free set before \globpart{} begins
any partitioning). Assume $h > 1$. So $F_h = V \setminus \bigcup_{v \in V_{<h}} \cluster(v)$.

If $u \notin F_h$, then there exists $v \in V_{<h}$ such that $u \in \cluster(v)$. By construction
\cluster$(v)$ is contained in $\supp(\widehat{M}^t \vec{v})$ for some $t \leq \len$. Thus, $v \in \ball(u)$
and $h_v < h$. Hence, $v$ will be considered in \Step{union} and the union will contain $u$. The output is NO.
For the converse, observe that if the output is NO, then there is a $v \in V_{< h}$ such that
$u \in \cluster(v)$. Hence, $u \notin F_h$.

Now for the running time analysis. The running time of \findball$(v)$ is $O(d^2\len^3\minp^{-2})$
(\Clm{findball}) and $|C| \leq \len\minp^{-1}$ (\Clm{ballsize}). Each call to \cluster{} takes
$O(d\len\minp^{-2})$ (\Clm{cluster}). The total running time is $O((d\len\minp^{-1})^3)$.

\end{proof}

We have the necessary tools to define the procedure \findr. We will need the following definition
in our description and analysis of \findr.

\begin{definition} \label{def:viable} Assume $\free{h} \geq \sizefrac n$. A vertex $s \in V_{\geq h}$ is called \emph{$(h,\sizethresh)$-viable}
if $C := \cluster(s,t_s,\sizethresh)$ is not a singleton and $|C \cap \free{h}| \geq \sizefrac^3 \sizethresh$.
(If $\free{h} < \sizefrac n$, no vertex is $(h,\sizethresh)$-viable.)
\end{definition}

Let us motivate this definition. When $C := \cluster(s,t_s,\sizethresh)$ is not a singleton, it is a low conductance
cut of $\Theta(\sizethresh)$ vertices. The vertex $s$ is $(h,\sizethresh)$-viable if $C$ contains a non-trivial fraction
of free vertices available in the $h$th phase. The viable vertices are those from which clustering will
make significant ``progress" in the $h$th phase. For each $h$, the procedure \findr{} searches for values of $\sizethresh$
that lead to many $(h,k)$-viable vertices. In the next section, we prove that having sufficiently many clusters come
from viable vertices ensures the cut bound of \Def{oracle}.

\medskip
\noindent\fbox{
\begin{minipage}{0.9\textwidth}
    {\findr$(\bR)$}
    \smallskip
    \begin{compactenum}
    \item For $h = 1$ to $\overline{h}$:
    \begin{compactenum}
        \item Sample $\sizefrac^{-10}$ uar vertices independently. Let $S_h$ be the multiset of sampled vertices
        that are in phase $\geq h$.
        \item If $|S_h| \leq \sizefrac^{-9}/2$, set $k_h = 0$ and continue for loop. Else, reset $S_h$
        to the multiset of the first $\sizefrac^{-8}$ vertices sampled.
        \item For $k \in [\minp^{-1}]$ and for every $s \in S_h$: \label{step:loop}
        \begin{compactenum}
            \item Compute $C := \cluster(s,t_s,k)$.
            \item For all $u \in C$, call \isfree$(u,h,K_{h-1})$ to determine if $u \in F_{h-1}$.
            \item If $C$ is not a singleton and $|C \cap F_{h-1}| \geq \sizefrac^3 \sizethresh$, mark
            $s$ as being $(h,\sizethresh)$-viable.
        \end{compactenum}
        \item If there exists some $\sizethresh$ such that there are at least $12\sizefrac^{4}|S_h|$ $(h,k)$-viable vertices, assign an arbitrary
        such $\sizethresh$ as $\sizethresh_h$. Else, assign $\sizethresh_h := 0$. \label{step:setthresh}
    \end{compactenum}
    \item Output $\ST_{\overline{h}} = \{\sizethresh_1, \sizethresh_2, \ldots, \sizethresh_{\overline{h}}\}$.
    \end{compactenum}
\end{minipage}
}

\medskip

\begin{claim} \label{clm:findr-time} The running time of \findr{} is $O((d\len\delta^{-1}\minp^{-1})^5)$.
\end{claim}

\begin{proof} There are $\overline{h} = \maxphase$ iterations. We compute the running time of each iteration. 
There are at most $\minp^{-1}\sizefrac^{-8}$ calls to \cluster, each of which takes $O(d\len\minp^{-2})$ time by \Clm{cluster}.
For each call to \cluster, there are at most $\minp^{-1}$ calls to \isfree. Each call to \isfree{} takes $O((d\len\minp^{-1})^3)$ time (\Clm{isfree}).
The running time of each iteration is $O(\sizefrac^{-10} + d\len\minp^{-3}\sizefrac^{-8} + d^3\len^3\minp^{-5}\sizefrac^{-8})$.
By the parameter settings, since $\len^2 \geq \eps^{2\cdot\lenexp} \geq (\sizefracval)^{-8} = \sizefrac^{-8}$, the running
time of each iteration $O((d\len\minp^{-1})^5)$. The total running time is $O((d\len\delta^{-1}\minp^{-1})^5)$.
\end{proof}

The following theorem gives the main guarantee of \findr. The proof is a fairly straightforward Chernoff bound
on top of an application of \Thm{restrict-cut}. Quite simply, the proof just says the following. \Thm{restrict-cut}
shows the existence of $(h,k)$ pairs for which many vertices are viable. The \findr{} procedure finds such
pairs by random sampling.

\begin{theorem} \label{thm:findr} The following property of the values $\ST_{\overline{h}}$ 
of \findr$(\bR)$ and the preprocessing choices holds with probability at least $1-\exp(-1/\eps)$ over
all the randomness in $\bR$.  For all $h \leq \overline{h}$, if $|\free{h}| \geq \sizefrac n$, at least $\sizefrac^5 \delta n$ vertices in $V_{h}$ are $(h,k_h)$-viable.
\end{theorem}

\begin{proof} The proof has two parts. In the first part, we argue that whp, if $|\free{h}| \geq \sizefrac n$,
then a non-zero $\sizethresh_h$ is output. This part is an application of \Thm{restrict-cut}. In the second part,
we prove that (whp), if a non-zero $\sizethresh_h$ is output, then it satisfies the desired properties. This part
is proven using a simple Chernoff bound argument.

Fix an $h$. Condition on any choice of $V_1, V_2, \ldots, V_{h-1}$ such that $|\free{h}| \geq \sizefrac n$.
Note that $V_{\geq h} \supseteq \free{h}$, since all vertices in $V_{< h}$ are necessarily clustered by
the $h$th phase. (Recall that \cluster$(v)$ always contains $v$.)
Hence, $|V_{\geq h}| \geq \sizefrac n$. There will be numerous low probability ``bad" events that we need to track. We will describe these bad events,
and refer to their probabilities as ``Error 1", ``Error 2", etc.

{\bf Error 1, $\exp(-\sizefrac^{-8})$.} The probability that a uar vertex is in $V_{\geq h}$ is at least $\sizefrac$,
and the expected size of $S_h$ is at least $\sizefrac \times \sizefrac^{-10} = \sizefrac^{-9}$.
By the Chernoff bound of \Thm{chernoff}, $\Pr[|S_h| \leq  \sizefrac^{-9}/2] \leq \exp(-\sizefrac^{-9}/12)$ $\leq \exp(-\sizefrac^{-8})$.
Thus, with probability at least $1-\exp(-\sizefrac^{-8})$, \Step{loop} is reached and $S_h$ is a multiset 
of iid uar $\sizefrac^{-8}$ elements in $V_{\geq h}$.

Let us assume that $S_h$ is such a multiset, and prove that a non-zero $\sizethresh_h$ is output whp. 
We bring out the main tool, \Thm{restrict-cut}. Since $|\free{h}| \geq \sizefrac n$,
there exists a size threshold $k \leq \minp$ such that the following holds. For at least $(\sizefrac^2/\log^2\sizefrac^{-1})n$
vertices $s \in \free{h}$, there are at least $(\sizefrac/\log^2\sizefrac^{-1})\len$ timesteps $t$ such that:
there exists $k' \in [k,2k]$ such that (i) $\level{s}{t}{k'} \subseteq \supp(\widehat{M}^t\vec{s})$, (ii) $\Phi(\level{s}{t}{k'} \cup \{s\}) < \phi$, and (iii) $|\level{s}{t}{k'} \cap F| \geq \sizefrac^3 k$.
For any such $(s, t, k)$ triple, consider a call to \cluster$(s,t,k)$. Observe that the call will output the largest
level set of size in $[k,2k]$ satisfying (i) and (ii).
Hence, it will output (non-singleton) $\level{s}{t}{k''}$ such that $k' \leq k'' \leq 2k$
and (i) and (ii) hold. Note that $\level{s}{t}{k''} \supseteq \level{s}{t}{k'}$, so the third item will also hold.
Thus, \emph{if $t_s$ is set to one of these $(\sizefrac/\log^2\sizefrac^{-1})\len$ timesteps $t$}, then $s$ will be $(h,k)$-viable.

{\bf Error 2, $\exp(-\sizefrac^{-1})$.} Let us fix a size threshold $k$ promised by \Thm{restrict-cut}. The probability that a uar element on $V_{\geq h}$
is marked as $(k,h)$-viable is at least the product of probability of choosing an appropriate $s$ with
the probability that $t_s$ is chosen appropriately. Thus, the probability of find an $(h,k)$-viable
vertex is at least $(\sizefrac^2/\log^2\sizefrac^{-1}) \times (\sizefrac/\log^2\sizefrac^{-1}) = \sizefrac^3/\log^4\sizefrac^{-1}$. 
This probability is independent for all vertices in $V_{\geq h}$. By the Chernoff bound in \Thm{chernoff},
with probability at least $1-\exp(-\sizefrac^4 |S_h|/12)$, at least $\sizefrac^3 |S_h|/2\log^4\sizefrac^{-1} \geq 12\sizefrac^4|S_h|$
$(h,k)$-viable vertices are discovered in \findr. In this case, in \Step{setthresh}, $\sizethresh_h$ is set to a non-zero value.
The probability of this event happening is at least $1-\exp(-\sizefrac^{-8})-\exp(-\sizefrac^4|S_h|/8)$ $\geq 1 - \exp(\sizefrac^{-1})$.
(Recall that whp $S_h$ is a multiset of iid uar $\sizefrac^{-8}$ vertices.
In the union bound above, the first ``bad event" is $S_h$ \emph{not} having $\sizefrac^{-8}$ vertices and the second ``bad event"
is discovering too few viable vertices.) We have concluded that whp, if $|\free{h}| \geq \sizefrac n$,
then $\sizethresh_h$ is non-zero.

We move to the second part of the proof, which asserts that (with high probability),
an output non-zero $\sizethresh_h$ has the desired properties. Condition on any choice of the preprocessing. 
Note that the randomness is only over the choice of $S_h$.
Fix any $\sizethresh \leq \minp^{-1}$.
Suppose that the number of $(h,\sizethresh)$-viable vertices in $V_{\geq h}$ is at most $2\sizefrac^5n$.
Then, the expected number of such vertices in $S_h$ is at most $2\sizefrac^5 n/|V_{\geq h}| \times |S_h|
\leq 2\sizefrac^4 |S_h|$. (We use the lower bound $|V_{\geq h}| \geq |\free{h}| \geq \sizefrac n$.)

{\bf Error 3, $2^{-12\sizefrac^{-4}}$.} Let $X_\sizethresh$ denote the random variable of the number of $(h,\sizethresh)$-viable vertices in $S_h$.
Since $X_\sizethresh$ is distributed
as a binomial, by the Chernoff bound of \Thm{chernoff}, $\Pr[X_\sizethresh > 12\sizefrac^4|S_h|] \leq 2^{-12\sizefrac^4|S_h|}$.
Note than when $X_\sizethresh < 12\sizefrac^4|S_h|$, then $\sizethresh_h$ cannot be $\sizethresh$.
All in all, for any $h$, any choice of the $t_v$s, and any choice of $\sizethresh$, if \Step{loop} is reached
and the number of $(h,\sizethresh)$-viable vertices in $V_{\geq h}$ is at most $2\sizefrac^5n$,
then $\sizethresh_h \neq \sizethresh$ with probability at least $1-2^{-12\sizefrac^{-4}}$. Taking the contrapositive,
if $\sizethresh_h \neq 0$ (\Step{loop} must have been reached), then the number of $(h,\sizethresh_h)$-viable
vertices in $V_{\geq h}$ is at least $2\sizefrac^5n$. 

{\bf Error 4, $2\exp(-\delta\sizefrac^5n/12)$.} Suppose the number of $(h,\sizethresh_h)$-viable
vertices in $V_{\geq h}$ is at least $2\sizefrac^5n$ . By \Clm{vh-size} applied on the set of $(h,\sizethresh_h)$-viable
vertices in $V_{\geq h}$, with probability at least $1-2\exp(-\delta\sizefrac^5n/12)$, the 
number of such viable vertices in $V_h$ is at least $\delta\sizefrac^5n$.

We take a union bound over the $\maxphase$ values of $h$, the $\minp^{-1}$ values of $k$,
and all errors encountered thus far. The total error probability is at most
$\maxphase\cdot\minp^{-1}(\exp(-\sizefrac^{-8}) + \exp(\sizefrac^{-1}) + 2^{-12\sizefrac^{-4}} + 2\exp(-\delta\sizefrac^5n/12))$.
Note that $\maxphase, \sizefrac, \minp^{-1}$ are $\poly(\eps^{-1})$, and thus the total error probability is at most $\exp(-\eps^{-1})$.
With the remaining probability, the following holds. For all phases $h$, if $|\free{h}| \geq \sizefrac n$,
a non-zero $\sizethresh_h$ is output. If a non-zero $\sizethresh_h$ is output, the number of $(h,\sizethresh_h)$-viable
vertices in $V_{h}$ is at least $\delta\sizefrac^5n$.
\end{proof}

\section{Proving the cut bound: the amortization argument} \label{sec:amort}

We come to the final piece of proving the guarantees of \Thm{main-intro}. We need to prove
that the number of edges cut by the partition of \globpart{} is at most $\eps nd$. This requires
an amortization argument explained below. For the sake of exposition, we will ignore
constant factors in this high-level description. One of the important takeaways is how
various parameters are chosen to prove the cut bound.

Consider phase $h$ where $|\free{h}| \geq \sizefrac n$. Let us upper bound the number
of edges cut by the clustering done on this phase. Roughly speaking, $|V_h| = \delta n$, so
there are $\delta n$ clusters created in this phase. Each cluster in this 
phase has at most $2\sizethresh_h$ vertices. The number of edges
cut by each such cluster is at most $2\phi \sizethresh_hd$ (since \cluster{} outputs a low conductance cut; ignore singleton outputs).
So the total number of edges cut is at most $2\phi \delta \sizethresh_h nd$.

Let us now lower bound the number of new vertices that are partitioned in phase $h$; this
is the set $\free{h+1} \setminus \free{h}$. For each $(h,\sizethresh_h)$-viable $v$ in $V_h$, \cluster$(v)$ contains at least $\sizefrac^3 \sizethresh_h$
vertices in $\free{h}$. These will be newly partitioned vertices. Here comes the primary difficulty:
the clusters for the different such $v$ might not be disjoint. We need to lower bound the union
of the clustered vertices in $\free{h}$. An alternate description of the challenge is as follows. We are only
guaranteed that clusters from viable vertices $v$ contains many vertices in $\free{h}$, the free set at the \emph{beginning}
of phase $h$. What we really need is for the cluster from $v$ to contain many free vertices \emph{at the time}
that $v$ is processed. Phases were introduced to solve this problem. By reducing $\delta$, we can limit the size of $V_h$,
thereby limiting the intersection between the clusters produced in this phase. 

We now explain the math behind this argument. Consider some $w \in \free{h}$ and let $c_w$ be the number
of vertices in $V_{\geq h}$ that cluster $v$ (call these seeds). Thus, $c_w = |\{s \ | \ s \in V_{\geq h}, v \in \cluster(s)\}$. 
The vertex $w$ is clustered in phase $h$ iff one of these $c_w$ seeds is selected in $V_h$. By \Clm{geo-phase},
each such seed is independently selected in $V_h$ with probability $\delta$. The probability
that $w$ is clustered in this phases is precisely $1-(1-\delta)^{c_w}$. Crucially, $c_w \leq |\ball(w)| \leq \len\minp^{-1}$.
We chose $\delta \ll \len\minp^{-1}$, so $1-(1-\delta)^{c_w} \approx \delta c_w$.

Thus, the expected number of newly clustered vertices is at least $\sum_{w \in \free{h}} \delta c_w$. 
By rearranging summations, $\sum_{w \in \free{h}} c_w = \sum_{v \in V_{\geq h}} |\cluster(v) \cap \free{h}|$.
For every $(h,\sizethresh_h)$-viable vertex $v$ in $V_{\geq h}$, $|\cluster(v) \cap \free{h}| \geq \sizefrac^3 \sizethresh_h$.
The arguments in the proof of \Thm{findr} shows that there are $\sizefrac^5 n$ such vertices in $V_{\geq h}$ whp. Hence, we can lower bound (in expectation) 
the new number of newly clustered vertices as follows: 
$$\sum_{w \in \free{h}} \delta c_w \geq \delta\cdot(\sizefrac^5 n)\cdot (\sizefrac^3 \sizethresh_h) = \delta\sizefrac^8 \sizethresh_h n$$

We upper bounded the number of edges cut by $2\phi \delta \sizethresh_h nd$. The ratio of edges cut to vertices clustered
is $8\phi \sizefrac^{-8}d$. The parameters are set to ensure that $8\phi \sizefrac^{-8} \ll \eps$, so the total number
of edges cut is $\eps nd$.

The formal analysis requires some care to deal with conditional probabilities and dependencies between
various phases. Also, \Thm{findr} talks about $V_h$ and not $V_{\geq h}$, which necessitates some changes. But the essence of the argument is the same.

Our main theorem is a cut bound for \globpart.

\begin{theorem} \label{thm:edge-cut} The expected number of edges cut by the partitioning
of \globpart$(\bR)$ is at most $\eps nd$.
\end{theorem}

We will break up the proof into two technical claims.
Somewhat abusing notation, we say a vertex in $V_{\geq h}$ is $h$-viable if it is $(h,\sizethresh_h)$-viable.

\begin{claim} \label{clm:cut-cluster} 
$$ \EX[\textrm{\# edges cut by \globpart$(\bR)$}] \leq 32 \phi\sizefrac^{-8}d^2 \Big(\sum_{h < \overline{h}} \EX[\sum_{v \in V_h} |\cluster(v) \cap \free{h}|)]\Big) + 
2\sizefrac nd $$
\end{claim}

\begin{proof} The proof goes phase by phase. We call a phase significant if $|\free{h}| \geq \sizefrac n$.
Edges cut in a significant phase are also called significant. Observe that the total number of edges cut
is at most the number of significant edges cut plus $\sizefrac nd$. (This contributes to the extra additive term in the claim statement.) 
Below, we will bound the total number of significant edges cut. 

By \Clm{last-phase}, with probability at least $1-2^{-\delta n}$, $|V_{\overline{h}}| \leq \delta n$.
Note that $|\free{h}| \leq V_{\geq \overline{h}} = |V_{\overline{h}}|$. (The equality is because this is the last phase.)
Since $\delta n < \sizefrac n$, the expected number of significant edges cut in the last phase
is at most $2^{-\delta n} nd < 1$.

Now assume that $h < \overline{h}$.
Consider the edges cut in the $h$th phase. Consider any choice of $V_1, V_2, \ldots, V_{h-1}$
and $\sizethresh_1, \sizethresh_2, \ldots, \sizethresh_{h}$. If $|\free{h}| < \sizefrac n$, no significant edges are cut.
Let us assume that $|\free{h}| \geq \sizefrac n$.
Each set \cluster$(v)$ output in this phase is either a singleton
or a set of size at most $2\sizethresh_h$ and conductance at most $\phi$. In either
case, the number of edges cut by removing \cluster$(v) \cap F$ (in \globpart) is at most
$2\phi \sizethresh_h d + d$. Note that $2\phi \sizethresh_h d \geq 1$ (otherwise, by the connectedness of $G$, there
can never be a set of size at most $2\sizethresh_h$ of conductance $\leq \phi$).
Hence, the number of significant edges cut by a single cluster is at most $2\phi \sizethresh_h(d+d^2) \leq 4\phi\sizethresh_h d^2$.

Note that
$|V_{\geq h}| \geq |\free{h}| \geq \sizefrac n$ and $|V_{\geq h}|$ is obviously at most $n$.
By \Clm{vh-size} with $S = V_{\geq h}$, with probability at least $1-2\exp(-\delta \sizefrac n/12)$ over the choice of $V_h$,
$|V_h| \leq 2\delta n$. Hence, the total number of significant edges cut is at most $4\phi\sizethresh_h d^2 \times 2\delta n = 8\phi\delta \sizethresh_h d^2n$.

By \Thm{findr}, with probability at least $1-\exp(\eps^{-1})$, if $|\free{h}| \geq \sizefrac n$, at least $\sizefrac^5 \delta n$ vertices
in $V_h$ are $h$-viable. Call this event $\cE$. For every $h$-viable vertex in $V_h$, $|\cluster(v) \cap \free{h}| \geq \sizefrac^3 \sizethresh_h$.
For convenience, let $X_h := \sum_{v \in V_h} |\cluster(v) \cap \free{h}|)$.
Conditioned on $\cE$, 
$X_h \geq \sizefrac^8 (\delta \sizethresh_h n)$. Recall that with probability at least $1-2\exp(-\delta\sizefrac n/12)$,
the number of significant edges cut in this phase is at most $8\phi d^2(\delta \sizethresh_h n)$.
If $\cE$ occurs, we can apply the bound $\sizefrac^{-8}X_h \geq \delta \sizethresh_h n$
and upper bound the number of significant edges cut in this phase by $8\phi\sizefrac^{-8}d^2 X_h$,

Thus, with probability at least $1-\exp(\eps^{-1})-2\exp(-\delta\sizefrac n/12)$,
the number of significant edges cut in phase $h$ is at most $(8\phi\sizefrac^{-8}d^2)X_h$. In other words, there is an event $\cF_h$
conditioned on which the above bound happens, and $\Pr[\cF_h] \geq 1-\exp(\eps^{-1})-2\exp(-\delta\sizefrac n/12)$. In the calculation
below, we break into conditional expectations and use the fact that $\delta = \poly(\eps)$, $\beta = \Theta(\eps)$, and
that the number of phases is at most $\maxphase$. We also use the fact that $X_h$ is non-negative.
\begin{eqnarray}
    & & \sum_h \EX[\textrm{\# significant edges cut in phase $h$}] \leq \sum_h (\Pr[\cF] \EX[X_h | \cF] + \Pr[\overline{\cF}] nd) \\
    & \leq & \sum_h \EX[X_h] + \maxphase(\exp(\eps^{-1})+2\exp(-\delta\sizefrac n/12))nd \leq \sum_h \EX[X_h] + \sizefrac nd/2
\end{eqnarray}
To this bound, we add the expected number of edges cut in the last phase (at most $1$) and the number
of non-significant edges cut (at most $\sizefrac n$). This completes the proof.
\end{proof}

\begin{claim} \label{clm:charging} 
$$ \sum_{h < \overline{h}} \EX[\sum_{{v \in V_h}} |\cluster(v) \cap \free{h}|)] \leq 4n $$
\end{claim}

\begin{proof} We will apply the following charging argument. When a vertex $v$ is processed in \globpart$(\bR)$,
we will add one unit of charge to every vertex in $\cluster(v) \cap \free{h}$. Note that the total amount
of charge is exactly the quantity we wish to bound. Crucially, note that any vertex $w$ receives charge
in at most one phase; the phase where it leaves the free set.

We will prove that the expected charge that any vertex receives is at most 4 units, which will prove
the claim. Fix a vertex $w$. Let $\chi$ be the random variable denoting the charge that $w$ receives,
and $\cE_{h}$ be the event that $w$ receives charge in phase $h$. Since $w$ receives charge in exactly 
one phase, $\EX[\chi] = \sum_h \EX[\chi | \cE_h] \Pr[\cE_h]$. We will prove that, for all $h$, $\EX[\chi | \cE_h] \leq 4$,
which implies that $\EX[\chi] \leq 4$ as desired.

To analyze $\EX[\chi | \cE_h]$, first condition on a setting of $V_1, V_2, \ldots, V_{h-1}$ (such that $w \in \free{h}$) and all other preprocessing
for all vertices. We refer to this setting as the event $\cC$. The randomness for specifying $V_h$ has not been set. The event $\cE_h$ occurs if there is a $v \in V_h$ such that
$w \in \cluster(v)$. The charge $\chi$ is the number of vertices $v \in V_h$ such that $w \in \cluster(v)$. 
Let $c$ be the number of such vertices in $V_{\geq h}$. 
Note that $v \in \ball(w)$, and by \Clm{ballsize}, $c \leq \len\minp^{-1}$.
%
%

By \Clm{geo-phase}, every vertex in $V_{\geq h}$ is in $V_h$ with probability $\delta$. Hence,
$\Pr[\cE_h | \cC] = 1-(1-\delta)^c$. Note that $\delta c \leq \delta \len \minp^{-1}
= (d^{-70+6+60}\eps^{\delexp} \cdot \eps^{\lenexp} \cdot \eps^{-\minpexp} < 1/2$. 
Hence $(1-\delta)^c \leq 1-\delta c + (\delta c)^2 \leq 1-\delta c/2$
and $\Pr[\cE_h | \cC] \geq \delta c/2$.
Note that $\EX[\chi | \cC] = \sum_{b > 0} {c \choose b} \delta^b \leq \sum_{b > 0} (\delta c)^b \leq 2 \delta c$.
Observe that $\EX[(\chi | \cE_h) | \cC] \leq (2\delta c)/(\delta c/2) = 4$.

Note that the event $\cE_h$ can be partitioned according to the different $\cC$ events.
Hence $\EX[\chi | \cE_h] = \sum_{\cC} \EX[(\chi | \cE_h) | \cC] \Pr[\cC] \leq 4$.
Thus, the proof is completed.
\end{proof}

\Thm{edge-cut} follows by a direct application of these claims and plugging in the parameter values.

\begin{proof} (of \Thm{edge-cut}) By \Clm{cut-cluster} and \Clm{charging}, the expected number of edges
cut by \globpart$(\bR)$ is at most $128\phi\sizefrac^{-8}d\cdot nd + 2\sizefrac nd$. Plugging in the parameters
$\phi = d^{-1}\eps^{\phiexp}$, $\sizefrac = \sizefracval$, and noting that $\eps$ is sufficiently small,
the expectation is at most $\eps nd$.
\end{proof}

We can now wrap up the proof of \Thm{main-intro}, showing the existence of $(\eps, \poly(d\eps^{-1}))$-partition oracles
for minor-closed families.

\begin{proof} (of \Thm{main-intro}) The procedure for the partition oracle is \findpart$(v,\bR)$. Let us prove each property
of \Def{oracle}. 

Consistency: By \Thm{findpart},
the partition created by calls to \findpart$(v,\bR)$ is precisely the same as the partition
created by \globpart$(\bR)$. 

Cut bound: By \Thm{edge-cut}, the expected number
of edges cut is at most $\eps nd$. 

Running time: The running time of \findpart$(v,\bR)$ is $O((d\len\minp^{-1})^5)$ plus
the running time of \findr. The running time of \findr{} is $O((d\len\delta^{-1}\minp^{-1})^5)$, by 
\Clm{findr-time}. By the parameter settings, $\len, \delta^{-1}, \minp^{-1}$ are all $\poly(d\eps^{-1})$. Hence, the total running time of \findpart$(v,\bR)$ is also $\poly(d\eps^{-1})$.
\end{proof}

\section{Diffusion Behavior on Minor-Free Families}\label{sec:diffusion}

In this section, we state and prove the main theorem about diffusions on minor-free graph classes.
This is the (only) part of the paper where the property minor-freeness makes an appearance. \Thm{goodseed}
is used in the proof of \Sec{ls-cluster}. For convenience, we recall the parameters involved.

\parameterinfo

\begin{theorem} \label{thm:goodseed} Let $G$ be a bounded degree graph in minor-closed family. Let $F$ be an arbitrary subset
of at least $\sizefrac n$ vertices. 
There are at least $\sizefrac^2 n/8$ vertices $s \in F$ such that: for at least $\sizefrac\len/8$
timesteps $t \in [\len]$, $\trwalk^t\vec{s}(F) \geq \sizefrac/16$.
\end{theorem}

We note that this theorem holds for all graphs, if we replace the truncated walk $\trwalk$
by the standard random walk $M$. The main insight is that, for $G$ in a minor-closed family, ``polynomial" truncation
of the walk distribution does not significantly affect the behavior.

The main property of bounded degree minor-free graphs we require is hyperfiniteness,
as expressed by Proposition 4.1 of~\cite{AST90} (also used as Lemma 3.3 of~\cite{KSS:19}).

\begin{theorem} \label{thm:ast} There is an absolute constant $\gamma$ such that the following holds.
Let $H$ be a graph on $r$ vertices. Suppose $G$ is an $H$-minor-free graph. Then, for all $b \in \NN$,
there exists a set of at most $\gamma r^{3/2}n/\sqrt{b}$ vertices whose removal leaves $G$ with
all connected components of size at most $k$.
\end{theorem}

The key stepping stone to proving \Thm{goodseed} is \Lem{walk}, which shows that truncation does not 
affect walk distributions from many vertices. Let us first state a simple fact on $l_1$-norms.

\begin{fact} \label{fact:norm} Let $\vec{x}$ and $\vec{y}$ be vectors with non-negative entries, such
that for all coordinates $i$, $\vec{x}(i) \geq \vec{y}(i)$.
Then $\|\vec{x} - \vec{y}\|_1 = \|\vec{x}\|_1 - \|\vec{y}\|_1$.
\end{fact}

\begin{proof} $\|\vec{x} - \vec{y}\|_1 \geq \sum_i |\vec{x}(i) - \vec{y}(i)| = \sum_i (\vec{x}(i) - \vec{y}(i)) =
\|\vec{x}\|_1 - \|\vec{y}\|_1$.
\end{proof}

This fact bears relevance for us, since truncations of walk distribution vectors only reduce coordinates.

\begin{lemma} \label{lem:walk} For at least $(1-\minp^{1/8})n$ vertices $v$, the following holds. 
For every $t \leq \len$, $\| M^t \vec{v} - \trwalk^t \vec{v}\|_1 \leq \len \minp^{1/9}$.
\end{lemma}

\begin{proof} Let $H$ be an arbitrary forbidden minor for the minor-closed family of interest.
We first apply \Thm{ast} with $k = \lceil 1/\sqrt{\minp} \rceil$. There exists
a set $C$ of at most $\gamma r^{3/2} \minp^{1/4} dn$ edges who removal leads to connected
components of size at most $\lceil 1/\sqrt{\minp} \rceil \leq 2/\sqrt{\minp}$. For convenience,
set the constant $\gamma' := \gamma r^{3/2}$. We will need the following claim.

\begin{claim} \label{clm:trapped} For at least $(1-\minp^{1/8})n$ vertices $v$, the probability
that an $\len$-length random walk encounters an edge of $R$ is at most $\gamma' \len \minp^{1/8}$. 
\end{claim} 

\begin{proof}
The proof is a Markov bound argument. Suppose not; so there exist strictly more
than $\minp^{1/8}n$ vertices $v$ such that
an $\len$-length random walk encounters an edge of $C$ with at least $\gamma' \len \minp^{1/8}$ probability.
Consider an $\len$-length random walk that starts from the uniform (also stationary) distribution. The above
assumption implies that the expected number of $C$ edges encountered is $> \minp^{1/8} \cdot \gamma' \len \minp^{1/8} = \gamma' \len \minp^{1/4}$.
On the other hand, since the walk remains in the stationary distribution, for all $t \leq \len$, the probability of encountering
an edge in $C$ at the $t$th step is precisely $|C|/2dn$. (Recall that the lazy random walk has $1/2dn$ of taking any edge.)
By linearity of expectation, the expected number of $C$ edges encountered is $\len |C|/2dn$. By the bound of \Thm{ast}, 
$\len |C|/2dn \leq \gamma' \len \minp^{1/4}$ contradicting the bound obtained from the assumption.
\end{proof}

Consider such a vertex $v$, as promised by the previous paragraph. Let $S$ be the connected component over vertices that contains $v$,
after removing the edge cut $C$. Let $q_t$ be the probability that the walk from $v$ leaves $S$
at the $t$th step; by the property of the previous parameter, $\sum_{t \leq \len} p_t \leq \gamma' \len \minp^{1/8}$.
Let $M_S$ be the transition matrix
of the random walk $M$ \emph{restricted to $S$}. Note that $M_S$ is not necessarily stochastic.
We will use the truncated walk $\trwalk_S$. Observe that $\|\trwalk^t \vec{v}\|_1 \geq \|\trwalk^t_S \vec{v}\|_1$.

Since all coordinates of $\trwalk^t \vec{v}$ are at most those of $M^t \vec{v}$, by \Fact{norm},
$\|M^t \vec{v} - \trwalk^t \vec{v}\|_1 = \|M^t \vec{v}\|_1 - \|\trwalk^t \vec{v}\|_1$. 
Since $\|M^t \vec{v}\|_1 = 1 = \|\vec{v}\|_1$ and $\|\trwalk^t \vec{v}\|_1 \geq \|\trwalk^t_S \vec{v}\|_1$,
we can upper bound as follows by a telescoping sum.
\begin{eqnarray} 
\|M^t \vec{v} - \trwalk^t \vec{v}\|_1 & \leq & \sum_{l = 1}^t \Big(\|\trwalk^{l-1}_S\vec{v}\|_1 - \|\trwalk^{l}_S \vec{v}\|_1\Big) \\
& = & \sum_{l=1}^t \Big(\|\trwalk^{l-1}_S\vec{v}\|_1 - \|M_S\trwalk^{l-1}_S\vec{v}\|_1 + \|M_S\trwalk^{l-1}_S\vec{v}\|_1 - \|\trwalk^{l}_S \vec{v}\|_1\Big)\label{eq:norm-diff}
\end{eqnarray}

The quantity $\|\trwalk^{l-1}_S\vec{v}\|_1 - \|M_S\trwalk^{l-1}_S\vec{v}\|_1$ is exactly the probability
that a single step (according to $M$) from $\trwalk^{l-1}_S\vec{v}$ leaves $S$. Since all coordinates in $\trwalk^{l-1}_S\vec{v}$ are
at most those of $M^{l-1}\vec{v}$, this probability is at most $q_l$. 
The quantity $\|M_S\trwalk^{l-1}_S\vec{v}\|_1 - \|\trwalk^{l}_S \vec{v}\|_1$ is the probability mass
lost by truncation of $M_S\trwalk^{l-1}_S\vec{v}$. We apply the 
trivial bound $\minp|S|$. This is where the hyperfiniteness plays a role; since $|S| \leq 2/\sqrt{\minp}$,
$\|\trwalk_S \vec{x}_{l-1} - M_S \vec{x}_{l-1}\|_1 \leq \minp\cdot/2\sqrt{\minp} = 2\sqrt{\minp}$. 

We sum all the these bounds over $l \leq t$, and plug into \Eqn{norm-diff}. We bound
$\|M^t \vec{v} - \trwalk^t \vec{v}\|_1 \leq \sum_{l \leq t} p_l + 2t \sqrt{\minp}$.
By the properties of $v$, this is at most $\gamma'\len \minp^{1/8} + 2\len\sqrt{\minp} \leq \len\minp^{1/9}$
(for sufficiently small $\minp$).
%
%
\end{proof}

We are now ready to prove \Thm{goodseed}.
We will need the following simple ``reverse Markov" inequality for bounded random variables.
\begin{fact} \label{fact:markov} Let $X$ be a random variable taking values in $[0,1]$ such
that $\EX[X] \geq \delta$. Then $\Pr[X \geq \delta/2] \geq \delta/2$.
\end{fact}

\begin{proof} Let $p$ be the probability that $\Pr[X \geq \delta/2]$.
\begin{eqnarray*}
 \delta \leq \EX[X] & = & \Pr[X \geq \delta/2] \EX[X | X \geq \delta/2] + \Pr[X < \delta/2] \EX[X | X < \delta/2] \\
 & \leq & p + (1-p)(\delta/2) \leq p + \delta/2
\end{eqnarray*}
\end{proof}

%
%
%
%
%
%
 
\begin{proof} (of \Thm{goodseed}) Define $\theta_{s,t}$ as follows. For $s \in F$ and $t \in [\len]$:
if $t$ is odd, $\theta_{s,t} = 0$. If $t$ is even, then $\theta_{s,t}$
is the probability that the $t$-length random walk starting from $s$ ends in $F$.

Let us pick a uar source vertex in $s \in F$, pick a uar length $t \in [\len]$.
We use the fact
that $M$ is a symmetric matrix. We use $\bone_F$ to denote the all $1$s vector on $F$.
\begin{eqnarray}
\EX_{s,t}[\theta_{s,t}]= \bone^T_F \sum_{i=1}^{\len/2} (M^{2i}/\len) (\bone_F/|F|) = (\len|F|)^{-1} \sum_{i \leq \len/2} \bone^T_F M^{2i} \bone_F = (\len|F|)^{-1} \sum_{i \leq \len/2} \|M^i \bone_F\|^2_2 \label{eq:tau}
\end{eqnarray}
Note that $\|M^i\bone_F\|_1 = |F|$, so by Jensen's inequality, $ \|M^i \bone_F\|^2_2 \geq |F|^2/n$.
Plugging in \Eqn{tau}, $\EX_{s,t}[\theta_{s,t}] \geq \len^{-1}\times (\len/2)|F|/n \geq \sizefrac/2$.
For any $s$, $\EX_t[\theta_{s,t}] \leq 1$. By \Fact{markov}, there are at least $\sizefrac|F|/4$
vertices $s \in F$ such that $\EX_t[\theta_{s,t}] \geq \sizefrac/4$. Again applying \Fact{markov},
for at least $\sizefrac|F|/4$ vertices $s \in F$, there are at least $\sizefrac \len/8$ timesteps
$t \in [\len]$ such that $\theta_{s,t} \geq \sizefrac/8$, implying that $M^t\vec{s}(F) \geq \sizefrac/8$.

By \Lem{walk}, there are at least $(1-\minp^{1/8})n$ vertices $s$ such that for all $t \leq \len$,
$\| M^t \vec{s} - \trwalk^t \vec{s}\|_1 \leq \len \minp^{1/9} = d^{6 - 60/9}\eps^{-\lenexp + \minpexp/9} \leq \sizefrac/16$.
By the parameters settings, $\minp^{1/8} < \eps^{\minpexp/8} \leq \sizefrac|F|/8$.
Invoking the bound from the previous paragraph, there are at least $\sizefrac|F|/8$
satisfying the property of \Lem{walk} and the condition at the end of the previous paragraph.
For all such vertices $s$, for all $t \leq \len$, $\trwalk^t\vec{s}(F) \geq M^t\vec{s}(F) - \sizefrac/16$.
Thus, for all such $s$, there are at least $\sizefrac\len/8$ timesteps $t \in [\len]$
such that $\trwalk^t\vec{s}(F) \geq \sizefrac/16$.

%
%
%
%
\end{proof}

\section{The proof of \Thm{findr}: local partitioning within $F$} \label{sec:ls-cluster}

We repeat the parameter values for convenience.

\parameterinfo


Recall that \Thm{restrict-cut} shows that there are many $s \in F$ from which 
(level sets of) diffusions in $G$ discover low conductances cuts in $F$.
We use the Lov\'{a}sz-Simonovits curve to represent the truncated diffusion vector,
and keep track of the vertices of $F$ wrt to the curve. 
This is done via a careful adaptation of Lov\'{a}sz-Simonovits method,
as presented in \Lem{flatten}.

The main technical tool which we will use in our analysis is the Lov\'{a}sz-Simonovits method,
defined in~\cite{LS:90}, whose use for clustering was pioneered by~\cite{ST12}.
    \begin{definition}
        For a non-negative vector $\mathbf{p}$ over $V$, the function $I: \R^n \times [n] \rightarrow [0, 1]$ is defined as
        $$ I(\mathbf{p}, x) =  \max_{\substack{\mathbf{w} \in [0, 1]^n \\ \sum \mathbf{w}(u) = x }}\sum_{u \in V} \mathbf{p}(u)\mathbf{w}(u)$$
        This is equivalent to summing over the $x$ heaviest elements of $\mathbf{p}$ when $x$ is an integer, and linearly interpolating between these points otherwise.

        For notational convenience, we define:
        $$I_{s, t}(x) = I(\trwalk^t \vec{s}, x)\textrm{.}$$
    \end{definition}

Note that $I_{s,t}$ is a concave curve.

\subsection{The Lov\'{a}sz-Simonovits lemma} \label{sec:ls}

The fundamental lemma of Lov\'{a}sz-Simonovits is the following (Lemma 1.4 of~\cite{LS:90}, also refer to Theorem 7.3.3 of Lecture 7 of~\cite{Sp-notes}).

\begin{lemma} \label{lem:ls} Let $\overline{x} = \min(x,n-x)$. 
	Consider any non-negative vector $\vec{p}$, and let $S_x$ denote the level
	set of $M\vec{p}$ with $x$ vertices.
	$$ I(M\vec{p}, x) \leq (1/2)(I(\vec{p},x-2\overline{x}\Phi(S_x)) + I(\vec{p},x-2\overline{x}\Phi(S_x))) $$
\end{lemma}

The concavity of the curves implies monotonicity, $I(M\vec{p}) \leq I(\vec{p})$.
The application of this lemma to our setting leads to the following statement.

\begin{lemma} \label{lem:lstrun} For all $t \leq \len$ and $x \leq 1/\minp$,
$$ \ls{s}{t}(x) \leq (1/2)(\ls{s}{t-1}(x(1-\Phi(\level{s}{t}{x}))) + \ls{s}{t-1}(x(1+\Phi(\level{s}{t}{x}))))$$
\end{lemma}

Let $\lin{t}{w}{y}$ be the straight line between the points $(w,\ls{s}{t}(w))$ and $(y,\ls{s}{t}(y))$.

\begin{lemma} \label{lem:flatten} Let $t_0 < t_1 < \ldots < t_h$ be time steps.
	Suppose $\forall i \leq h$ and $x \in [w,y]$: $\level{s}{t_i}{x} \subseteq \supp(\trwalk^t \vec{s})$  $\Longrightarrow$
	$\Phi(\level{s}{t_i}{x}) \geq \psi$.
	Then, $\forall i \leq h, \forall x \in [w,y]$
	$$ \ls{s}{t_i}(x) \leq \lin{t_0-1}{w}{y}(x) + \sqrt{\min(x-w, y-x)} (1-\psi^2/128)^i $$
\end{lemma}

\begin{proof} For convenience, let $\Delta_x = \min(x-w, y-x)$.
We prove by induction over $i$.

For showing the base case take $i=0$. Now consider the following cases.
\begin{asparaitem}
    \item Suppose $x = w$ or $x = y$. By monotonicity, $\ls{s}{t_0}(x) \leq \ls{s}{t_0-1}(x)$.
Since $x \in \{w,y\}$, the latter is exactly $\lin{t_0}{w}{y}(x)$.
    \item Suppose $x \in [w+1, y-1]$. Then $\Delta_x \geq 1$ and $\ls{s}{t_0}(x) \leq 1 \leq \sqrt{\Delta_x}$.
    \item Suppose $x \in (w,w+1)$. Note that $\Delta_x = w-x < 1$. By the definition of the LS curve,
    $\ls{s}{t_0}(x) = \ls{s}{t_0}(w) + (w-x)(\ls{s}{t_0}(w+1) - \ls{s}{t_0}(w)) $ $\leq \ls{s}{t_0-1}(w) + \sqrt{w-x}$
    $\leq \lin{t_0-1}{w}{y}(x) + \sqrt{\Delta_x}$.
    \item Suppose $x \in (y-1,y)$. An identical argument to the above holds.
\end{asparaitem}

Now for the induction. Suppose the premise holds at step $t_i$. Namely for $x \in [w,y]$, for 
all level sets $\level{s}{t_i}{x}$ contained inside $\supp(\trwalk^t \vec{s})$,
$\Phi(\level{s}{t_i}{x}) \geq \psi \geq \psi$. We would like to upperbound $\ls{s}{t_i}(x)$.
To this end, let us consider some $x \in [w,y]$. By
\Lem{lstrun},

\begin{eqnarray}
\ls{s}{t_i}(x) & \leq & (1/2)[\ls{s}{{t_i-1}}(x(1-\Phi(\level{s}{t_i}{x}))) + 
\ls{s}{{t_i-1}}(x(1+\Phi(\level{s}{t_i}{x})))] \\
& \leq & (1/2)[\ls{s}{{t_{i-1}}}(x(1-\Phi(\level{s}{t_i}{x}))) + 
\ls{s}{{t_{i-1}}}(x(1+\Phi(\level{s}{t_i}{x})))]
\end{eqnarray}

\noindent The second inequality follows by monotonicity, since $t_{i-1} \leq t_i - 1$. Note that $\Delta_x  = \min(x-w, y-x) \leq x$ for all $x \in [w,y]$. 
\Clm{step-by-step-drop} (which we prove after the current lemma) shows the following.

\begin{claim}\label{clm:step-by-step-drop}
	For all $1 \leq i \leq h$, for all $x \in [w,y]$, the following holds
	\begin{eqnarray}
		\ls{s}{t_i}(x) & \leq & (1/2)[\ls{s}{{t_{i-1}}}(x-\Delta_x\psi/4)) + 
	\ls{s}{{t_{i-1}}}(x+\Delta_x \psi/4))] \label{eq:ls-recur}
	\end{eqnarray}
\end{claim}

%
\noindent Now, let $x_L = x - \Delta_x \psi/4$ and $x_R = x + \Delta_x \psi/4$. Using \Clm{step-by-step-drop} we get

\begin{eqnarray}
\ls{s}{t_i}(x) & \leq & (1/2)[\lin{t_0-1}{w}{y}(x_L) + \sqrt{\Delta_{x_L}}
(1-\psi^2/128)^{i-1} \nonumber \\
& & + \lin{t_0-1}{w}{y}(x_R) + \sqrt{\Delta_{x_R}}(1-\psi^2/128)^{i-1}] \\
& = & (1/2)[\lin{t_0-1}{w}{y}(x_L) + \lin{t_0-1}{w}{y}(x_R)] \nonumber\\
& & + (1/2)[\sqrt{\Delta_{x_L}}) (1-\psi^2/8)^{i-1} + \sqrt{\Delta_{x_R}}(1-\psi^2/128)^{i-1}] \label{eq:ls}
\end{eqnarray}
\noindent Here, \Eqn{ls} follows from the induction hypothesis.
Since $\lin{t_0-1}{w}{y}$ is a linear function, the first term is exactly $\lin{t_0-1}{w}{y}(x)$.
We analyze the second term. 

We first assume that $\Delta_{x} = x-w$ (instead of $y-x$).
\begin{eqnarray}
\Delta_{x_L} & = & \min(x-\psi\Delta_{x}/4 - w, y-x+\psi\Delta_{x}/4) \\
& = & \min((1-\psi/4)\Delta_{x}, y-x+\psi/4\Delta_{x}) \leq (1-\psi/4)\Delta_{x}
\end{eqnarray}
Analogously, 
\begin{eqnarray}
\Delta_{x_R}  &= & \min(x+\psi\Delta_{x}/4 - w, y-x-\psi\Delta_{x}/4) \\
& = & \min((1+\psi/4)\Delta_{x}, y-x-\psi\Delta_{x}/4) \leq (1+\psi/4)\Delta_{x}
\end{eqnarray}
Thus, the second term of \Eqn{ls} is at most $(1/2)(1-\psi^2/128)^{i-1}\sqrt{\Delta_{x}}(\sqrt{1-\psi/4} + \sqrt{1+\psi/4})$.

Now, we consider $\Delta_{x} = y-x$.
\begin{eqnarray}
\Delta_{x_L} & = & \min(x-\psi\Delta_{x}/4 - w, y-x+\psi\Delta_{x}/4) \\
& = & \min(x-\psi\Delta_{x}/4 - w, (1+\psi/4)\Delta_{x}) \leq (1+\psi/4)\Delta_{x}
\end{eqnarray}
Analogously, 
\begin{eqnarray}
\Delta_{x_R}  &= & \min(x+\psi\Delta_{x}/4 - w, y-x-\psi\Delta_{x}/4) \\
& = & \min(x+\psi\Delta_{x}/4 - w, (1-\psi/4)\Delta_{x}) \leq (1-\psi/4)\Delta_{x}
\end{eqnarray}
In this case as well, the second term of \Eqn{ls} is at most $(1/2)(1-\psi^2/128)^{i-1}\sqrt{\Delta_{x}}(\sqrt{1-\psi/4} + \sqrt{1+\psi/4})$.

In both cases, we can upper bound \Eqn{ls} as follows. (We use the inequality
$\frac{\sqrt{1-z} + \sqrt{1+z}}{2} \leq 1-z^2/8$.
$$ \ls{s}{t_i}(x) \leq \lin{t_0-1}{w}{y}(x) + (1-\psi^2/128)^{i-1}\sqrt{\Delta_{x}}\frac{\sqrt{1-\psi/4} + \sqrt{1+\psi/4}}{2} \leq \lin{t_0-1}{w}{y}(x) + (1-\psi^2/128)^i\sqrt{\Delta_{x}}$$

\end{proof}

Now, we establish \Clm{step-by-step-drop}, the missing piece in the above proof.

\begin{proof} (of \Clm{step-by-step-drop})
Suppose $x_{max} \in [w,y]$
is the maximum value of $x \in [w,y]$ for which $\level{s}{t_i}{x}$ is still inside
the support of the truncated diffusion at the $t_i$-th step. We split into three cases:
$x \leq x_{max}$, $x \in (x_{max}, x_{max} + \Delta_{x_{max}}\psi/2]$, $x > x_{max} + \Delta_{x_{max}}\psi/2$.
Note that in the latter two cases, $\level{s}{t_i}{x}$ is not contained in $\supp(\trwalk^{t_i}\vec{s})$.

{\bf Case 1, $x \leq x_{max}$:} Note that \Eqn{ls-recur} holds by concavity of the Lov\'{a}sz-Simonovits curve when 
$\level{s}{t_i}{x} \subseteq \supp(\trwalk^{t_i} \vec{s})$ (because then this level set has 
conductance at least $\psi$). 

{\bf Case 2, $x \in (x_{max}, x_{max} + \Delta_{x_{max}}\psi/2]$:}
Let $S = \level{s}{t_i}{x_{max}}$ and
let $T = \level{s}{t_i}{x}$. Observe that

\begin{align}
	\Phi(T) = \frac{|E(T, \overline{T})|}{d|T|} 
	\stackrel{(\bone)}\geq \frac{|E(S,\overline{S})| - \psi/2 \cdot d|S|}{d|S| + \psi/2 \cdot d|S|} 
	\stackrel{(\btwo)}\geq \frac{\psi d|S|/2 }{2d|S|} 
	\geq \frac{\psi}{4}
\end{align}

Here, $(\bone)$ follows because $T$ could contain at most $\psi |S|/2 $ neighbors of $S$ which
could cost us at most $\psi d|S|/2$ edges in the cut $(S, \overline{S})$. $(\btwo)$ follows by
upperbounding $\psi$ by $1$. Again the claim in \Eqn{ls-recur} follows by concavity of the
Lov\'{a}zs-Simonovits curve. \\

{\bf Case 3, $x > x_{max} + \Delta_{x_{max}}\psi/2$:}
Now let $x_r = x_{max} + \Delta_{x_{max}} \psi/2$. 
Write $x = x_{max} + \Delta_{x_{max}} \psi/2 + s$. Recall $\Delta_x = \min(x-w, y-x)$.
We claim that $x - \Delta_x \psi/4 \geq x_{max}$. First let us see how to establish \Eqn{ls-recur}
assuming this claim holds. Assuming this claim, we have
$$\ls{s}{t_i}(x - \Delta_x \psi/4) = \ls{s}{t_i}(x_{max}) = \ls{s}{t_i}(x + \Delta_x \psi/4) = \|\trwalk^{t_i} \vec{s}\|_1.$$
And therefore,

\begin{align*}
	\ls{s}{t_i}(x) &= \frac{1}{2} \cdot \left[\ls{s}{t_i}(x - \Delta_x \psi/4) + \ls{s}{t_i}(x + \Delta_x \psi/4) \right] \\
	&\leq \frac{1}{2} \cdot \left[\ls{s}{t_{i-1}}(x - \Delta_x \psi/4) + \ls{s}{t_{i-1}}(x + \Delta_x \psi/4)\right]
\end{align*}

Now, all that remains to establish \Eqn{ls-recur} is to show $x - \Delta_x \psi/4 \geq x_{max}$. For simplicity, write
$\Delta_m = \Delta_{x_{max}}$. Now consider two cases depending on the value of $\Delta_m$
\begin{enumerate}
	\item {\bf Case 1} $\Delta_m = x_{max} - w$.

		In this case note that
		\begin{align*}
			x - \Delta_x \psi/4 &= x_{max} + \Delta_m \psi/2 + s - (x - w) \psi/4 \\
			&\geq x_{max} + \Delta_m \psi/2 + s - (x_{max} + \Delta_m \psi/2 + s - w) \psi/4 \\
			&\geq x_{max} + \Delta_m \psi/4 - \Delta_m \psi^2/8 + s - s \psi/4 \\
			&\geq x_{max} + \Delta_m \psi/8 + s(1 - \psi/4) \geq x_{max}
		\end{align*}

		which establishes the claim above as desired.

	\item {\bf Case 2} $\Delta_m = y - x_{max}$.

		In this case note that
		\begin{align*}
			x - \Delta_x \psi/4 &= x_{max} + \Delta_m \psi/2 + s - (y - x) \psi/4 \\
			&\geq x_{max} + \Delta_m \psi/2 + s - (y - x_{max} - \Delta_m \psi/2 - s) \psi/4 \\
			&\geq x_{max} + \Delta_m \psi/4 + \Delta_m \psi^2/8 + s + s \psi/4 \\
			&\geq x_{max}
		\end{align*}
\end{enumerate}

Thus, in both cases, the claim from above holds. This means that \Eqn{ls-recur} holds as long as the
premise holds for the $t_i$-th step.
\end{proof}

\subsection{From leaking timesteps to the dropping of the LS curve} \label{sec:leak}

We fix a source vertex $s$, and consider the evolution of $\trwalk^t\vec{s}$.
Therefore, we drop the dependence of $s$ from much of the notation.

We use $\prw{t}$ to denote $\trwalk^t\vec{s}$. We begin with a 
few definitions.

\begin{definition} \label{def:leaking} A timestep $t$ is called \emph{leaking for source $s$}
if, for all $k \leq \minp^{-1}$: if $\level{s}{t}{k} \subseteq \supp(\trwalk^{t}\vec{s})$ and $|\level{s}{t}{k} \cap F| \geq \alpha^2 k/400$,
then $\Phi(\level{s}{t}{k}) \geq 1/d \len^{1/3}$.

If timestep $t$ is not leaking for $s$, there exists $k \leq \minp^{-1}$ such that 
$\level{s}{t}{k} \subseteq \supp(\trwalk^{t}\vec{s})$, $|\level{s}{t}{k} \cap F| \geq \alpha^2 k/400$, and $\phi(\level{s}{t}{k}) < 1/d \len^{1/3}$. Such
a $k$ is denoted as an \emph{$(s,t)$-certificate of non-leakiness}. 
\end{definition}

We set $\alpha = \eps^{4/3}/300,000$.

Following the construction of the LS curve $\ls{s}{t}$, we will order each
vector $\prw{t}$ in decreasing order, breaking ties by id.
The \emph{rank} of a vertex is its position in (the sorted version of) $\prw{t}$.

\begin{definition} \label{def:bucket} Let the \emph{bucket} $\bucket{t}{r}$ denote
the set of vertices whose rank in $\prw{t}$ is in the range $[2^r, 2^{r+1})$.

A bucket $\bucket{t}{r}$ is called \emph{heavy} if $\sum_{v \in \bucket{t}{r} \cap F} \prw{t}(v) \geq \alpha$. (The bucket restricted to $F$ has large probability.)
\end{definition}

The following lemma says that if there are many leaking timesteps, then the LS curve drops at heavy buckets.

\begin{lemma} \label{lem:drop} Fix $r \geq 0$. Suppose  for some $s \in F$, there exist 
	$\len' \geq \sizefrac^3 \len/8$ leaking timesteps $t_0 < t_1 < \ldots < t_{\len'}$
	such that for all $0 \leq i \leq \len'$, $\bucket{t_i}{r}$ is heavy. 
	Then, $\ls{s}{t_{\len'}}(2^{r+1}) < \ls{s}{t_0}(2^{r+1}) - \alpha/4$.
\end{lemma}

The main tool used in our 
proof is our adaptation of Lov\'{a}sz-Simonovits lemma done in \Lem{flatten}. We first
make a definition.

\begin{definition} \label{def:balanced:split}
	Fix $r \geq 0$, a source $s$ and a timestep $t$. A vertex $w \in [2^r, 2^{r+1}]$ is called 
	a balanced split for $t$ if $|\lev{t}{w} \cap F| \geq \alpha 2^r/3$
	and $\sum_{v \in \bucket{t}{r} \setminus \lev{t}{w}} \prw{t}(v) \geq \alpha/3$.
\end{definition}

We will first prove the following claim which essentially follows by averaging arguments.

\begin{claim} \label{clm:split} 
	Fix $r \geq 0$ and suppose for some source vertex $s \in F$, there exist
	$\len'$ leaking timesteps $t_0 < t_1 < \ldots < t_{\len'}$
	such that for all $0 \leq i \leq \len'$, $\bucket{t_i}{r}$ is heavy. 
	Then, there exists a vertex $w$ that is a balanced split for at least
	an $\alpha/3$-fraction of timesteps in $T = \{t_0, t_1, \ldots t_{\len'} \}$.
\end{claim}

\begin{proof} 
Since $\bucket{t_0}{r}$ is heavy, $\ls{s}{t_0}(2^r) < 1$. Since the support
of $\prw{t}$ is at most $\minp^{-1}$, this implies that $2^r < \minp^{-1}$ and 
$r \leq -\lg \minp$ (and this holds by the choice of parameters).

For all $v \in \bucket{t}{r}$, $\prw{t}(v) \leq 1/2^r$.
Since $\sum_{v \in \bucket{t}{r} \cap F} \prw{t}(v) \geq \alpha$, $|\bucket{t}{r} \cap F| \geq \alpha 2^r$.

For convenince, let $T = \{t_0, t_1, \ldots t_{\len'}\}$. Pick $w$ uar in $[2^r,2^{r+1})$. Let $X_i$
be the indicator for $w$ being a balanced split for $t_i$. Recall that
$|\bucket{t_i}{r} \cap F| \geq \alpha 2^r$. Sort the vertices of $\bucket{t_i}{r} \cap F$
by increasing rank and consider the vertices in positions $\alpha 2^r/3$ and $2\alpha 2^r/3]$.
Let the rank corresponding to these vertices by $u_1$ and $u_2$.
We first argue that any rank $w \in [u_1, u_2]$ is a balanced split.
We have $|\lev{t}{w} \cap F| \geq \alpha 2^r/3$ because $w \geq u_1$.
For all $v \in \bucket{t_i}{r}$, $\prw{t_i}(v) \leq 1/2^r$. Thus,
$\sum_{v \in \lev{t_i}{u_2} \cap \bucket{t_i}{r}} \prw{t_i}(v) \leq (1/2^r)(2\alpha 2^r/3) = 2\alpha/3$.
Note that $\sum_{v \in \bucket{t_i}{r}} \prw{t}(v) \geq \alpha$, since the
bucket is heavy
Hence, for any $w \leq u_2$, 
$\sum_{v \in \bucket{t}{r} \setminus \lev{t}{w}} \prw{t}(v) \geq \alpha - 2\alpha/3 = \alpha/3$.

As a consequence, for any $t_i$, there are at least $\alpha 2^r/3$ values of $w$
that are balanced splits. In other words,
$\EX[X_i] \geq \alpha/3$. By linearity of expectation, $\EX[\sum_{i \leq \len'}X_i] \geq \alpha {\len'}/3$.
Thus, there must exist some $w \in [2^r, 2^{r+1})$ that is a balanced split for at least
$\alpha \len'/3$ timesteps.
\end{proof}

Next, we show the following claim which essentially uses leakiness of a timestep $t \in T$ and
the balanced split vertex $w$ promised by \Clm{split} to spell out a set with enough free vertices
with large conductance.

\begin{claim}\label{clm:large:conductance}
	Fix $r \geq 0$ and let $w \in [2^r, 2^{r+1})$ be a split vertex as promised by \Clm{split} and let 
	$t_{i_1} < t_{i_2} < \ldots < t_{i_{\alpha \len'/3}}$ denote the timesteps for which $w$ is a balanced split.
	Let $y = \min(2^{r+6+\lceil \lg(1/\alpha)\rceil}, \minp^{-1})$. 
	Then, for all $x \in [w,y]$ and for all $t \in \{t_{i_1}, t_{i_2}, \cdots, t_{i_{\alpha \len'/3}} \}$, 
	whenever $\lev{t}{x} \subseteq \supp(\trwalk^t \vec{s})$, then $\Phi(\lev{t}{x}) \geq 1/d \len^{1/3}$. 
\end{claim}

\begin{proof}
	Take $x \in [w,y]$ and a leaking timestep $t \in \{t_{i_1}, t_{i_2}, \cdots, t_{i_{\alpha \len'/3}} \}.$
	Note that $x \leq y \leq \minp^{-1}$ clearly holds. Now, to establish the lower bound on conductance 
	claimed, we first unpack what it 
	means for $t$ to be a leaking timestep \Def{leaking}. It says: 
	If $\lev{t}{x} \subseteq \supp(\trwalk^t \vec{s})$ and $|\lev{t}{x} \cap F| \geq \alpha^2 k/400$, 
	then it better hold that $\phi(\lev{t}{x}) \geq 1/d \len^{1/3}$.

	Note that $y \leq 2^{r+6+\lceil \lg(1/\alpha)\rceil} \in [2^r (64/\alpha), 2^{r+1}(64/\alpha)]$. 
	Since $r \leq -\lg \minp$, $y \leq 128(\minp\alpha)^{-1}$. 

	Note that for all $t \in \{t_{i_1}, t_{i_2}, \cdots, t_{i_{\alpha \len'/3}} \}$ 
	and $x \in [w,y]$, $\lev{t}{x}$ contains at least $\alpha 2^r/3$ vertices of $F$. 
	Thus, at least a $(\alpha 2^r/3)/(2^{r+1} \cdot 64/\alpha) \geq \alpha^2/400$-fraction 
	of $\lev{t}{x}$ is in $F$.
	Now note that since $t$ is leaking, we see that one of the following will hold. Either

	\begin{asparaitem}
		\item $\lev{t}{x} \subseteq \supp(\trwalk^t \vec{s})$ and $\Phi(\lev{t}{x}) \geq 1/d \len^{1/3}$, Or
		\item $\lev{t}{x} \not\subseteq \supp(\trwalk^t \vec{s})$.
	\end{asparaitem}

	And this establishes the claim.
\end{proof}

Now, we have all the ingredients to prove \Lem{drop}. The key step which remains is an application of \Lem{flatten}.

\begin{proof} (Of \Lem{drop})
Suppose $w \in [2^r, 2^{r+1})$ is a balanced split at $\alpha \len'/3$ timesteps as promised by
\Clm{split}. Let $y = \min(2^{r+6+\lceil \lg(1/\alpha) \rceil}, \minp^{-1})$ and as observed in
\Clm{large:conductance}, note that for $x \in [w,y]$ if $\lev{t}{x} \subseteq \supp(\trwalk^t \vec{s})$,
it holds that $\phi(\lev{t}{x}) \geq 1/d \len^{1/3}$.
Now, we apply \Lem{flatten}. For all $x \in [w,y]$, we have
$\ls{s}{t_{\len'}}(x) \leq \ls{s}{t_{i_{\alpha \len'/3}}}(x)
\leq \lin{t_{i_1-1}}{w}{y}(x) + \sqrt{x} (1-1/128 d^2\len^{2/3})^{\alpha \len'/3}$.
By the premise, $\len' \geq \sizefrac^3 \len/8$ and therefore we
have 
$$(1 - 1/128 d^2 \len^{2/3})^{\alpha \len'/3} \leq (1 - 1/128 d^2\len^{2/3})^{\alpha \sizefrac^3 \len/3} = \ 
(1 - 1/128 d^2 \len^{2/3})^{128 d^2\len^{2/3} \cdot \frac{\alpha \sizefrac^3 \len^{1/3}}{3 \cdot 128d^2}} \leq \exp(-1/\alpha)$$ which holds because, for sufficiently small $\eps > 0$, we have 
$$\len^{1/3} = \frac{d^2}{\eps^{10}} \geq \frac{d^2 \cdot 10^{20}}{\eps^7} \geq \frac{d^2}{\alpha^3 \beta^3}.$$
Further, by the monotonicity of LS curves,
$\ls{s}{t_{\len'}}(x) \leq \lin{t_{i_1-1}}{w}{y}(x) + \exp(-1/\alpha)$
$\leq \lin{t_{i_0}}{w}{y}(x) + \exp(-1/\alpha)$.
Specifically, we get
\begin{equation} \label{eq:lhs-lem-drop}
	\ls{s}{t_{\len'}}(2^{r+1}) \leq \lin{t_{i_0}}{w}{y}(2^{r+1}) + \exp(-1/\alpha).
\end{equation}

Since $w$ is a good split, $\ls{s}{t_{i_0}}(2^{r+1}) \geq \ls{s}{t_{i_0}}(w) + \alpha/3$.
Note that 
\begin{eqnarray}
\lin{t_{i_0}}{w}{y}(2^{r+1})  & = & \ls{s}{t_{i_0}}(w) + 
(2^{r+1}-w)\left(\frac{\ls{s}{t_{i_0}}(y) - \ls{s}{t_{i_0}}(w)}{y-w}\right) \nonumber \\
& \leq & \ls{s}{t_{i_0}}(w) + 2^{r+1}/(y/2) \\
&\leq& \ls{s}{t_{i_0}}(w) + 2^{r+1} \times \left(\frac{2 \alpha}{2^r \cdot 64} \right) = \ls{s}{t_{i_0}}(w) + \alpha/16
\end{eqnarray}

The first inequality above follows by upper bounding $\ls{s}{t_{i_0}}(y) - \ls{s}{t_{i_0}}(w)$
by $1$, dropping the negative term and noting that $y-w \geq y/2$ for a sufficiently small 
$\alpha$. Together with \Eqn{lhs-lem-drop}, we get

\begin{eqnarray}
\ls{s}{t_{\len'}}(2^{r+1}) \leq \lin{t_{i_0}}{w}{y}(2^{r+1}) + \exp(-1/\alpha)
& \leq & \ls{s}{t_{i_0}}(w) + \alpha/16 + \exp(-1/\alpha) \nonumber \\
&\leq & \ls{s}{t_{i_0}}(2^{r+1}) - \alpha/3
+\alpha/16 + \exp(-1/\alpha) 
\end{eqnarray}

By monotonicity of the LS curve, 
$\ls{s}{t_{\len'}}(2^{r+1}) < \ls{s}{t_0}(2^{r+1}) - \alpha/4$.

\end{proof}

Now, we state a key lemma. It says that a fixed bucket (parameterized by $r$)
satisfies the following at most timesteps: (i) either it does not contain enough free vertices, 
or (ii) if it contains many free vertices at a particular timestep, then most of the corresponding
timesteps are not leaky.

\begin{lemma}  \label{lem:heavy} Fix $r \geq 0$ and take any $s \in F$. There are at most $\sizefrac^3 \len/\alpha$ 
leaking timesteps $t$ (with respect to $s$) where $\bucket{t}{r}$ is heavy.
\end{lemma}

\begin{proof} We prove by contradiction. Suppose there are more than 
$\sizefrac^3 \len/\alpha$ leaking timesteps $t$ where $\bucket{t}{r}$ is heavy. 
We break these up into $4/\alpha$
contiguous blocks of $\sizefrac^3 \len / 4$ leaking timesteps. 
By \Lem{drop}, after every such block of timesteps,
$\ls{s}{t}(2^{r+1})$ reduces by more than $\alpha/4$. Note that $\ls{s}{0}(2^{r+1}) \leq 1$,
and thus, after $4/\alpha$ blocks, $\ls{s}{t}(2^{r+1})$ becomes negative. Contradiction 
to the non-negativity of $\ls{s}{t}(2^{r+1})$.
\end{proof}

\subsection{Proof of \Thm{restrict-cut}} \label{sec:relevant}

We finally prove \Thm{restrict-cut}. In particular, recall that this theorem claims that for an arbitrary
set $F \subseteq V$ with $|F| \geq \sizefrac n$, there exists a size threshold $k$ such that one
can find enough source vertices $s \in F$ such that $\len$-step diffusions from $s$ contain enough
non-leaky timesteps. Moreover, these non-leaky timesteps can be used to obtain a low conductance cut
restricted to $F$. We begin by showing that indeed many sources $s \in F$ have the desired behavior.

\begin{lemma} \label{lem:relevant} There are at least $\sizefrac^2n/8$ vertices $s \in F$,
such that: there are at least $\sizefrac\len/16$
timesteps $t$ in $[\len]$ that are not leaking for $s$.
\end{lemma}

\begin{proof} We fix any vertex $s$ satisfying the conditions of \Thm{goodseed}.
Let us recall what this means. This means that for at least $\sizefrac \len/8$ timesteps
$t$, it holds that $\trwalk^t \vec{s}(F) \geq \sizefrac/16$.
We will show that conclusion in \Lem{relevant} above holds for $s$ which will
establish the lemma. We prove by contradiction. 

To this end, let us suppose for any vertex $s$ satisfying the conditions 
of \Thm{goodseed}, there are at most $\sizefrac\len/16$ non-leaky timesteps. 
There are at least $\sizefrac\len/8-\sizefrac\len/16 = \sizefrac\len/16$ 
timesteps $t$ that are leaking for $s$, such that 
$\trwalk^t\vec{s}(F) \geq \sizefrac/16$. Fix any such timestep $t$ and consider
the buckets $\bucket{t}{r}$. There are at most $-\lg \minp$ buckets with non-zero
probability mass, and by averaging,
there exists $r \leq -\lg\minp$ such that 
$$\sum_{v \in F \cap \bucket{t}{r}} \prw{t}(v) \geq \sizefrac/(-16\lg \minp) = \frac{\eps}{160 \cdot 3000 \lg(1/\eps)} \geq \frac{\eps^{4/3}}{300,000} = \alpha$$
where the last step holds for sufficiently small $\eps$
and therefore, $\bucket{t}{r}$ is heavy.

Thus, for each of the $\sizefrac\len/16$ leaking timesteps $t$ above, 
there exists some $r \leq -\lg\minp$ such that $\bucket{t}{r}$ is heavy.
By averaging, there exists some $r \leq -\lg\minp$ such that
for $\sizefrac\len/(-16\lg \minp)$ leaking timesteps $t$, $\bucket{t}{r}$ is heavy.
However, for sufficiently small $\eps$ ($\eps < 2^{-30}$), 
we have 
$$\frac{\sizefrac\len}{-16 \lg \minp} = \frac{\eps \cdot \len}{160 \cdot 3000 \log(1/\eps)} \geq 1000 \eps^{3-4/3} \len \geq \frac{\sizefrac^3 \len}{\alpha}$$ 
which contradicts \Lem{heavy}.
\end{proof}

\begin{lemma} \label{lem:goodr} Let $|F| \geq \sizefrac n$. There exists a $r \leq \lg(1/\minp)$
such that for $\geq \sizefrac^2 n/(8 \lg^2(\minp^{-1}))$ vertices $s \in F$, the following holds.
For at least $\sizefrac \len/(\lg^2(\minp^{-1}))$ timesteps $t$,
there exists $k \in [2^r, 2^{r+1}]$ that is an $(s,t)$-certificate of non-leakiness.
\end{lemma}

\begin{proof} This is an averaging argument. Apply \Lem{relevant}.
For each of the $\sizefrac^2n/8$ vertices $s \in F$, there are at least $\sizefrac\len/16$
timesteps $t$ that are not leaking for $s$. Thus, for every such $(s,t)$ pair,
there exists $k_{s,t} \leq \minp^{-1}$ that is an $(s,t)$-certificate of non-leakiness.
We basically bin the logarithm of the certificates. Thus, to every pair $(s,t)$
(of the above form), we associate $r_{s,t} = \lfloor \lg k_{s,t} \rfloor$. By averaging,
for each relevant $s$, there is a value $r_s$ such that for at least $\sizefrac\len/(16\lg(\minp^{-1}))$ 
timesteps $t$, there is an $(s,t)$-certificate in $[2^{r_s}, 2^{r_s+1}]$.
Again, by averaging there exists $r \leq \lg(\minp{-1})$ such that
there are at least $\sizefrac^2 n/(8 \lg(\minp^{-1})) \geq \sizefrac^2 n/(\lg^2(\minp^{-1}))$ 
vertices $s \in F$ for which
there exist at least $\sizefrac\len/(16 \lg(\minp^{-1})) \geq \sizefrac\len/\lg^2(\minp^{-1})$ timesteps $t$,
such that there is an $(s,t)$-certificate for non-leakiness in $[2^r, 2^{r+1}]$.
\end{proof}

\Thm{restrict-cut} follows as a corollary of \Lem{goodr}. We now present the proof.

\begin{proof} (Of \Thm{restrict-cut})
	As seen from \Lem{goodr}, there exists some $r \leq -\lg(\minp)$ such that
	there are at least $\Omega(\sizefrac^2/\lg(\sizefrac^{-1})) \cdot n$ vertices
	$s \in F$ each of which in turn has $(s,t)$-certificates of non-leakiness for 
	at least $\Omega(\sizefrac/16 \lg^2(\sizefrac^{-1})) \cdot \len$ different values
	of $t$. We simply choose $k = 2^r$.

	Let $S \subseteq F$ denote the collection of these relevant sources. And for
	$s \in S$, define 
	$$C_s = \{t \leq \len : \text{ there exists a } (s,t)-\text{ certificate of non-leakiness} \}.$$
	Take $s \in S$, $t \in C_s$. We will show that there exists $k' = k'(s,t) \in [k,2k]$ 
	such that the level set $\level{s}{t}{k'}$ satisfies the following.

	\begin{asparaitem}
		\item $\level{s}{t}{k} \subseteq \supp(\widehat{M}^t\vec{s})$. 
		\item $\phi(\level{s}{t}{k'} \cup \{s\}) \leq 1/ \len^{1/3}$.
		\item $|\level{s}{t}{k'} \cap F| \geq \alpha^2 k'/400 \geq \sizefrac^3 k$.
	\end{asparaitem}

	The first item above follows from the conclusion of \Lem{goodr}, \Def{leaking} and taking 
	contrapositive in \Lem{flatten}.
	Unpacking, this means that since $t \in C_s$ is a non-leaking timestep for $s$, it follows that
	there exists $k' = k'(s,t) \in [k,2k]$ for which 
	$\level{s}{t}{k'} \subseteq \supp(\widehat{M}^t\vec{s})$.
	The last item above holds for this choice of $k'$ from the conclusion of \Lem{goodr}.
	For item 2 above, again note that our choice of $k'$ and \Lem{goodr} imply that 
	$$\phi(\level{s}{t}{k'}) \leq 1/d \len^{1/3} = 1/d \cdot \frac{\eps^{10}}{d^2} = \eps^{10}/d^3 = \phi/d^3$$ 
	and therefore $\phi(\level{s}{t}{k'} \cup \{s\}) \leq \phi$ also follows as by (possibly)
	including a single vertex in the set, the number of cut-edges can only increase by $d$.
\end{proof}

\section{Proofs of applications} \label{sec:appl}

The proofs here are quite straightforward and appear (in some form) in previous work.
We sketch the proofs, and do not give out the specifics of the Chernoff bound calculations.
Specifically, we mention Theorem 9.28 and its proof in ~\cite{G17-book}, which contains
these calculations.

\begin{proof} (of \Thm{testers}) 
	Given input graph $G$, we set up the partition oracle with proximity parameter $\eps/8$.
    Therefore, with probability at least $2/3$ over the random seed $\bR$, the number of cut edges is at most $\eps dn/8$.
    The tester repeats the following $O(1)$ times. For a random $\bR$, we first estimate the number of edges cut by random sampling. 
    The tester samples $\Theta(1/\eps)$ uar vertices $u$, picks a uar neighbor $v$ of $u$, and calls the partition
	oracle on $u$ and $v$. If these lie in different components, the edge $(u,v)$ is cut.
	If more that an $\eps/4$ fraction of edges are cut, then repeat with a new $\bR$.
    Otherwise, we fix the seed $\bR$ and proceed to the second phase of the tester.
    (If no such $\bR$ is found, the tester rejects.)

	In the second phase, we sample a multiset $S \subseteq V$ of $O(\eps^{-1})$ uar vertices, 
	and query the subgraph induced by the component $C(v)$ (of the partition given by the oracle)
	that each $v \in S$ belongs to. For each ($\poly(\eps^{-1})$-sized) component $C(v)$, we directly determine if it belongs to $\cQ$.
	(If there is an efficient algorithm, we can run that algorithm.) If any of these components
	does not belong to $\cQ$, the tester rejects, otherwise it accepts.

	Now, let us argue that this is a bonafide tester for $\cQ$. Recall $\cQ$ is both monotone and additive. 
    Suppose $G \in \cQ$. Since $\cQ$ is a subproperty of a minor-closed property, the first phase
    of setting the partition oracle succeeds with high probability. Since $\cQ$ is monotone and additive,
    all the subgraphs induced on the connected components $C(v)$ also satisfy $\cQ$. So the tester accepts
    whp. Suppose $G$ is $\eps$-far from $\cQ$. If the first phase does not succeed, then the tester rejects.
    So assume that the first phase succeeds. Whp, by a Chernoff bound, the number of cut edges (of the partition)
    is at most $\eps dn/2$. Since $\cQ$ is monotone, the graph obtained by removing these cut
    edges is at least $\eps/2$-far from $\cQ$. Since $\cQ$ is additive, at least $\Omega(\eps n)$ vertices
    participate in connected components that not in $\cQ$. Hence, by a Chernoff bound, the second phase
    rejects whp.

    The query complexity has at most an $O(d\eps^{-1})$ multiplicative overhead of the time complexity
    of the partition oracle, which is $\poly(d\eps^{-1})$. If $\cQ$ can be decided in polynomial time,
    then the second phase also runs in $\poly(d\eps^{-1})$ time.
\end{proof}

\begin{proof} (of \Thm{approx}) As with the previous proof, we set up the partition oracle with proximity parameter $\eps dn/c$,
	where $c$ is the largest amount by which an edge addition/deletion changes $f$.
    As before, there is a first phase to determine an appropriate setting of $\bR$ for the partition oracle.
	We sample $\poly(d\eps^{-1})$ uar vertices and determine the component that each
	vertex belongs to. For each component, we compute $f$ exactly. We take the sum of $f$-values,
	and rescale appropriately to get an additive $\eps nd$ estimate for $f$. 
\end{proof}

\section*{Acknowledgements} We acknowledge Reut Levi for pointing out a correction in the statement of
\Thm{testers}.

\bibliographystyle{alpha}
\bibliography{polytime-oracle}

\end{document}

%% file: preamble.tex
	\usepackage{a4,geometry}

\usepackage{graphicx}
\usepackage{amsmath,amssymb,amsthm,mathtools}
\usepackage{paralist}
\usepackage{longtable}
\usepackage{colortbl}
\usepackage{hhline}
\usepackage{bm}
\usepackage{xspace}
\usepackage{url}
\usepackage{fullpage, prettyref}
\usepackage{boxedminipage}
\usepackage{wrapfig}
\usepackage{ifthen}
\usepackage{color}
\usepackage{xcolor}
\usepackage{framed}
\usepackage{algorithmic,algorithm}
\usepackage[pagebackref,letterpaper=true,colorlinks=true,pdfpagemode=none,urlcolor=blue,linkcolor=blue,citecolor=violet,pdfstartview=FitH]{hyperref}
\usepackage{fullpage}

\usepackage{thmtools}
\usepackage{thm-restate}

\newtheorem{theorem}{Theorem}[section]

\newtheorem{lemma}[theorem]{Lemma}
\newtheorem{claim}[theorem]{Claim}
\newtheorem{corollary}[theorem]{Corollary}
\newtheorem{definition}[theorem]{Definition}

\newtheorem{fact}[theorem]{Fact}

\newcommand{\ignore}[1]{}

\DeclareMathOperator*{\supp}{supp}

\newcommand{\cC}{{\cal C}}

\newcommand{\cE}{{\cal E}}
\newcommand{\cF}{\mathcal{F}}

\newcommand{\cP}{\mathcal{P}}
\newcommand{\cQ}{\mathcal{Q}}

\newcommand{\R}{\mathbb R}

\newcommand{\eps}{\varepsilon}

\newcommand{\poly}{\mathrm{poly}}

\newcommand{\bone}{{\bf 1}}
\newcommand{\btwo}{{\bf 2}}
\newcommand{\bthree}{{\bf 3}}

\newcommand{\bA}{\boldsymbol{A}}

\newcommand{\bR}{\boldsymbol{R}}

\newcommand{\bP}{\boldsymbol{P}}

\newcommand{\NN}{\mathbb{N}}
\newcommand{\RR}{\mathbb{R}}

\newcommand{\EX}{\hbox{\bf E}}

\newcommand{\Sec}[1]{\hyperref[sec:#1]{Sec.\,\ref*{sec:#1}}} 
\newcommand{\Eqn}[1]{\hyperref[eq:#1]{(\ref*{eq:#1})}} 
\newcommand{\Fig}[1]{\hyperref[fig:#1]{Fig.\,\ref*{fig:#1}}} 
\newcommand{\Tab}[1]{\hyperref[tab:#1]{Tab.\,\ref*{tab:#1}}} 
\newcommand{\Thm}[1]{\hyperref[thm:#1]{Theorem\,\ref*{thm:#1}}} 
\newcommand{\Fact}[1]{\hyperref[fact:#1]{Fact\,\ref*{fact:#1}}} 
\newcommand{\Lem}[1]{\hyperref[lem:#1]{Lemma\,\ref*{lem:#1}}} 
\newcommand{\Prop}[1]{\hyperref[prop:#1]{Prop.~\ref*{prop:#1}}} 
\newcommand{\Cor}[1]{\hyperref[cor:#1]{Corollary~\ref*{cor:#1}}} 
\newcommand{\Conj}[1]{\hyperref[conj:#1]{Conjecture~\ref*{conj:#1}}} 
\newcommand{\Def}[1]{\hyperref[def:#1]{Definition~\ref*{def:#1}}} 
\newcommand{\Alg}[1]{\hyperref[alg:#1]{Alg.~\ref*{alg:#1}}} 
\newcommand{\Ex}[1]{\hyperref[ex:#1]{Ex.~\ref*{ex:#1}}} 
\newcommand{\Clm}[1]{\hyperref[clm:#1]{Claim~\ref*{clm:#1}}} 
\newcommand{\Obs}[1]{\hyperref[obs:#1]{Observation~\ref*{obs:#1}}} 
\newcommand{\Step}[1]{\hyperref[step:#1]{Step~\ref*{step:#1}}} 
\newcommand{\Assumption}[1]{\hyperref[assm:#1]{Assumption\,\ref*{assm:#1}}} 

\newcommand{\Sesh}[1]{{\color{red} Sesh: #1}}

%% file: intro.tex
\newpage
\setcounter{page}{1}

\section{Introduction} \label{sec:intro}

The algorithmic study of planar graphs is a fundamental direction in theoretical
computer science and graph theory. Classic results like the Kuratowski-Wagner characterization \cite{K30, W37},
linear time planarity algorithms \cite{HT74}, and the Lipton-Tarjan separator theorem underscore the significance
of planar graphs \cite{LiptonT:80}. The celebrated theory of Robertson-Seymour give a grand generalization of 
planar graphs through minor-closed families \cite{RS:12, RS:13, RS:20}. This has led to many deep results in
graph algorithms, and an important toolkit is provided by separator theorems and associated
decompositions \cite{AST:94}. 

Over the past decade, there have been many advances in \emph{sublinear} algorithms for
planar graphs and minor-closed families. We focus on the model
of random access to bounded degree adjacency lists, introduced by Goldreich-Ron~\cite{GR02}.
Let $G = (V,E)$ be a graph with vertex set $V = [n]$ and degree bound $d$. 
The graph is accessed through \emph{neighbor queries}: there is an oracle
that on input $v \in V$ and $i \in [d]$, returns the $i$th neighbor of $v$.
(If none exist, it returns $\bot$.)

One of the key properties of bounded-degree graphs in minor-closed families
is that they exhibit hyperfinite decompositions. A graph $G$ is hyperfinite
if $\forall \; 0 < \eps < 1$, one can remove $\eps dn$ edges from $G$
and obtain connected components of size independent of $n$ (we refer to these as pieces). For minor-closed families,
one can remove $\eps dn$ edges and get pieces of size $O(\eps^{-2})$.

The seminal result of Hassidim-Kelner-Nguyen-Onak (HKNO) \cite{HKNO} introduced the notion
of \emph{partition oracles}. This is a local procedure that provides ``constant-time"
access to a hyperfinite decomposition. The oracle takes a query vertex $v$ and outputs
the piece containing $v$. Each piece is of size independent of $n$, and at most $\eps dn$
edges go between pieces. Furthermore, all the answers are consistent with a single
hyperfinite decomposition, despite there being no preprocessing or explicit coordination. 
(All queries uses the same random seed, to ensure consistency.)
Partition oracles
are extremely powerful as they allow a constant time procedure to directly access
a hyperfinite decomposition. As observed in previous work, partition oracles lead
to a plethora of property testing results and sublinear time approximation algorithms for minor-closed graph families~\cite{HKNO,NS13}.
In some sense, one can think of partition oracles as a moral analogue of Sz\'{e}meredi's regularity lemma
for dense graph property testing: it is a decomposition tool that immediately yields a litany
of constant time (or constant query) algorithms. 

We give a formal definition of partition oracles. (We deviate somewhat from the definition in Chap. 9.5 of Goldreich's book~\cite{G17-book}
by including the running time as a parameter, instead of the set size.)

\begin{definition} \label{def:oracle} Let $\cP$ be a family of graphs with degree bound $d$ and $T: (0,1) \to \NN$ be a function. 
A procedure $\bA$ is an \emph{$(\eps,T(\eps))$-partition oracle} for 
$\cP$ if it satisfies the following properties. The deterministic procedure takes as input random access to $G = (V,E)$ in $\cP$,
random access to a random seed $r$ (of length polynomial in graph size), a proximity parameter $\eps > 0$, and a vertex $v$ of $G$.
(We will think of fixing $G, r, \eps$, so we use the notation $\bA_{G,r,\eps}$. All probabilities are with respect to $r$.)
The procedure $\bA_{G,r,\eps}(v)$ outputs a set of vertices and satisfies the following properties.
\begin{enumerate}
	\item (Consistency) The sets $\{\bA_{G,r,\eps}(v)\}$, over all $v$, form a partition of $V$. Also,
		these sets $\bA_{G,r,\eps}(v)$ induce connected graphs for all $v \in V$.
    \item (Cut bound) With probability (over $r$) at least $2/3$, the number of edges between the sets $\bA_{G,r,\eps}(v)$ is at most $\eps dn$.
    \item (Running time) For every $v$, $\bA_{G,r,\eps}(v)$ runs in time $T(\eps)$.
\end{enumerate}
\end{definition}

We stress that there is no explicit ``coordination" or sharing of state
between calls to $\bA_{G,r,\eps}(v)$ and $\bA_{G,r,\eps}(v')$ (for $v \neq v'$).
There is no global preprocessing step once the random seed is fixed. The consistency guarantee
holds with probability $1$.
Note that the running time $T(\eps)$ is clearly an upper bound on the size of the sets $\bA_{G,r,\eps}(v)$. 
For minor-closed families, one can convert any partition oracle to one that output sets of size $O(\eps^{-2})$
with a constant factor increase in the cut bound.
(refer to the end of Sec. 9.5 in~\cite{G17-book}).

The challenge in partition oracles is to bound the running time $T(\eps)$. HKNO gave a partition oracle
with running time $(d\eps^{-1})^{\poly(d\eps^{-1})}$. Levi-Ron \cite{LR15} built on the ideas from HKNO
and dramatically improved the bound to $(d\eps^{-1})^{\log (d\eps^{-1})}$. Yet,
for all minor-closed families, one can (in linear time) remove $\eps dn$ edges to get connected
components of size $O(\eps^{-2})$. HKNO raise the natural open question as to whether
$(\eps,\poly(d\eps^{-1}))$-partition oracles exist.

In this paper, we resolve this open problem.

\begin{theorem} \label{thm:main-intro} 
	Let $\cP$ be the set of $d$-bounded degree graphs
	in a minor-closed family. There is an $(\eps,\poly(d\eps^{-1}))$-partition oracle for $\cP$.
\end{theorem}

\subsection{Consequences} \label{sec:conseq}

As observed by HKNO and Newman-Sohler \cite{NS13}, partition 
oracles have many consequences for
property testing and sublinear algorithms. 

Recall the definition of property testers. Let $\cQ$ be a property of graphs with degree bound $d$.
The distance of $G$ to $\cQ$ is the minimum number of edge additions/removals required to 
make $G$ have $\cQ$, divided by $dn$. 
A property tester for $\cP$ is a randomized procedure that takes query access to an input
graph $G$ and a proximity parameter, $\eps > 0$.
If $G \in \cP$, the tester accepts with probability at least $2/3$. If the distance
of $G$ to $\cQ$ is at least $\eps$, the tester rejects with probability at least $2/3$.
We often measure the query complexity as well as time complexity of the tester.

A direct consequence of \Thm{main-intro} is an ``efficient" analogue (for monotone and additive properties)
of a theorem of Newman-Sohler
stating that all properties of hyperfinite graphs are testable. 
A graph property closed under vertex/edge removals is called \emph{monotone}.
A graph property closed under disjoint union of graphs is called \emph{additive}.

\begin{theorem} \label{thm:testers} Let $\cQ$ be any monotone and additive property of bounded degree
graphs of a minor-closed family. There exists a $\poly(d\eps^{-1})$-query tester
for $\cQ$.

If membership in $\cQ$ can be determined exactly in polynomial (in input size) time,
then $\cQ$ has $\poly(d\eps^{-1})$-time testers.
\end{theorem}

An appealing
consequence of \Thm{testers} is that the property of bipartite planar graphs
can be tested in $\poly(d\eps^{-1})$ time. For any fixed subgraph $H$,
the property of $H$-free planar graphs can be tested in the same time. And all of
these bounds hold for any minor-closed family.

As observed by Newman-Sohler, partition oracles give sublinear query algorithms
for any additive graph parameter that is ``robust" to edge changes. Again, \Thm{main-intro}
implies an efficient version for minor-closed families.

\begin{theorem} \label{thm:approx} Let $f$ be a real-valued function on graphs
that changes by $O(1)$ on edge addition/removals, and has the property
that $f(G_1 \cup G_2) = f(G_1) + f(G_2)$ for graphs $G_1, G_2$ that are not connected
to each other.

For any minor-closed family $\cP$, there is a randomized algorithm that, 
given $\eps > 0$ and $G \in \cP$,
outputs an additive $\eps n$-approximation to $f(G)$ and makes $\poly(d\eps^{-1})$ queries.
If $f$ can be computed exactly in polynomial time, then the above algorithm runs
in $\poly(d\eps^{-1})$ time.
\end{theorem}

The functions captured by \Thm{approx} are quite general. Functions such as maximum matching,
minimum vertex cover, maximum independent set, minimum dominating set, maxcut, etc. all
have the robustness property. As a compelling application of \Thm{approx}, we can get
$(1+\eps)$-approximations\footnote{The maximum matching is $\Omega(n/d)$ for 
a connected bounded degree graph. One simply sets $\eps \ll 1/d$ in \Thm{approx}.}
for the maximum matching in planar (or any minor-closed family) graphs in $\poly(d\eps^{-1})$ time.

These theorems are easy consequences of \Thm{main-intro}. Using the partition oracle,
an algorithm can essentially assume that the input is a collection of connected components
of size $\poly(d\eps^{-1})$, and run an exact algorithm on a collection of randomly sampled
components. We sketch the proofs in \Sec{appl}.

\subsection{Related work} \label{sec:related}

The subject of property testing and sublinear algorithms in bounded degree graphs
is a vast topic. We refer the reader to Chapters 9 and 10 of Goldreich's textbook \cite{G17-book}.
We focus on the literature relevant to sublinear algorithms for minor-closed families.

The first step towards a characterization of testable properties in the bounded-degree model
was given by Czumaj-Sohler-Shapira, who showed hereditary properties in non-expanding graphs
are testable \cite{CSS09}. This was an indication that notions like hyperfiniteness are connected
to property testing. Benjamini-Schramm-Shapira achieved a breakthrough by showing that all minor-closed
properties are testable, in time triply-exponential in $d\eps^{-1}$ \cite{BSS08}. Hassidim-Kelner-Nguyen-Onak introduced partition oracles, and designed one running in time $\exp(d\eps^{-1})$. Levi-Ron improved this bound to quasipolynomial in $d\eps^{-1}$, using a clever analysis inspired
by algorithms for minimum spanning trees~\cite{LR15}.  Newman-Sohler built on partition oracles
for minor-close families to show that all properties of hyperfinite graphs are testable~\cite{NS13}.
Fichtenberger-Peng-Sohler showed any testable property contains a hyperfinite property~\cite{FiPeSo19}.

There are two dominant combinatorial ideas in this line of work. The first is using subgraph
frequencies in neighborhood of radius $\poly(\eps^{-1})$ to characterize properties. This
naturally leads to exponential dependencies in $\poly(\eps^{-1})$. The second idea is to use
random edge contractions to reduce the graph size. Recursive applications lead to hyperfinite
decompositions, and the partition oracles of HKNO and Levi-Ron simulate this recursive
procedure. This is extremely non-trivial, and leads to a recursive local procedure
with a depth dependent of $\eps$. Levi-Ron do a careful simulation, ensuring that the recursion depth
is at most $\log(d\eps^{-1})$, but this simulation requires looking at 
neighborhoods of radius $\log(d\eps^{-1})$.
Following this approach, there is little hope of getting a recursion depth independent
of $\eps$, which is required for a $\poly(d\eps^{-1})$-time procedure.

Much of the driving force behind this work was the quest for a $\poly(d\eps^{-1})$-time tester
for planarity. This question was resolved recently using a different
approach from spectral graph theory, which was itself developed for sublinear time algorithms
for finding minors~\cite{KSS:18, KSS:19}. A major inspiration is the random walk
based one-sided bipartiteness tester of Goldreich-Ron \cite{GR99}.
This paper is a continuation of that line of work, and is a further demonstration
of the power of spectral techniques for sublinear algorithms. The tools build
on local graph partitioning techniques pioneered by Spielman-Teng \cite{ST12},
which is itself based on classic mixing time results of Lov\'{a}sz-Simonovits \cite{LS:90}. 
In this paper, we
develop new diffusion-based local partitioning tools that form the core of partition oracles.

We also mention other key results in the context of sublinear algorithms for minor-closed families,
notably the Czumaj et al \cite{C14} upper bound of $O(\sqrt n)$ for testing cycle minor-freeness, 
the Fichtenberger et al \cite{FLVW:17} upper bound of $O(n^{2/3})$ for testing $K_{2,r}$-minor-freeness,
and $\poly(d\eps^{-1})$ testers for outerplanarity and bounded treewidth graphs~\cite{YI:15,EHNO11}.

\section{Main Ideas} \label{sec:ideas}

The starting point for this work are the spectral methods used in~\cite{KSS:18,KSS:19}.
These methods discover cut properties within a neighborhood of radius $\poly(d\eps^{-1})$, without
explicitly constructing the entire neighborhood.

One of the key tools used in these results in a local partitioning algorithm, based on
techniques of Spielman-Teng~\cite{ST12}. The algorithm takes a seed vertex $s$,
performs a diffusion from $s$ (equivalently, performs many random walks) of length $\poly(d\eps^{-1})$,
and tracks the diffusion vector to detect a low conductance cut around $s$ in $\poly(d\eps^{-1})$
time. We will use the term \emph{diffusions}, instead of random walks, because we prefer
the deterministic picture of a unit of ``ink" spreading through the graph.
A key lemma in previous results states that, for graphs in minor-closed families,
 this procedure succeeds 
from more than $(1-\eps)n$ seed vertices. This yields a global algorithm to construct a 
hyperfinite decomposition with components of $\poly(d\eps^{-1})$ size. Pick a vertex
$s$ at random, run the local partitioning procedure to get a low conductance cut, remove
and recurse. Can there be a local implementation of this algorithm?

Let us introduce some setup. We will think of a global algorithm that processes
seed vertices in some order. Given each seed vertex $s$, a local partitioning algorithm
generates a low conductance set $C(s)$ containing $s$ (this is called a cluster). The final
output is the collection of these clusters. For any vertex $v$, let the \emph{anchor} of $v$
be the vertex $s$ such that $v \in C(s)$. A local implementation boils down to finding
the anchor of query vertex $v$. 

Observe that at any point of the global procedure, some vertices have been clustered,
while the remaining are still \emph{free}.
The global procedure described above seems hopeless for a local implementation. The cluster
$C(s)$ is generated by diffusion in some subgraph $G'$ of $G$, which was the set of free
vertices when seed $s$ was processed. Consider a local procedure trying to discover
the anchor of $v$. It would need to figure out the free set corresponding
to every potential anchor $s$, so that it can faithfully simulate the diffusion
used to cluster $v$. From an implementation standpoint, it seems
that the natural local algorithm is to use diffusions from $v$ in $G$
to discover the anchor.  But diffusion in a subgraph $G'$ is markedly different from $G$ and difficult
to simulate locally. Our first goal is to design a partitioning method using
diffusions directly in $G$. 

\medskip

{\bf Finding low conductance cuts in subsets, by diffusion in supersets:}
Let us now modify the global algorithm with this constraint in mind. At some stage of
the global algorithm, there is a set $F$ of free vertices. We need to find a low conductance
cut contained in $F$, while running random walks in $G$. Note that we must be able
to deal with $F$ as small as $O(\eps n)$. Thus, random walks (even starting from $F$)
will leave $F$ quite often; so how can these walks/diffusions find cuts in $F$?

One of our main insights is that these challenges can be dealt with, even for
diffusions of $\poly(d\eps^{-1})$ length. We show
that, for a uniform random vertex $s \in F$, a spectral partitioning algorithm
that performs diffusion from $s$ in $G$ can detect low conductance cuts contained in $F$.
Diffusion in the superset (all of $V$) provides information about the subset $F$.
This is a technical and non-trivial result, and crucially uses
the spectral properties of minor-closed families. Note that diffusions
from $F$ can spread very rapidly in short random walks, even in planar graphs.
Consider a graph $G$, where $F$ is a path on $\eps n$ vertices, and there
is a tree of size $1/\eps$ rooted at every vertex of $F$. Diffusions from any vertex
in $F$ will initially be dominated by the trees, and one has to diffuse for at least $1/\eps$
timesteps before structure within $F$ can be detected. Thus, the proof 
of our theorem has to look at average behavior
over a sufficiently large time horizon before low conductance cuts in $F$ are ``visible".
Remarkably, it suffices to look at $\poly(d\eps^{-1})$ timesteps to find structure in $F$,
because of the behavior of diffusions in minor-closed families.

The main technical tool used is the Lov\'{a}sz-Simonovits curve technique \cite{LS:90},
whose use was pioneered by Spielman-Teng \cite{ST12}. We also use the truncated
probability vector technique from Spielman-Teng to give cleaner implementations and proofs.
A benefit of using diffusion (instead of random walks) on truncated vectors 
is that the clustering becomes deterministic.

\medskip

{\bf The problem of ordering the seeds:} With one technical hurdle out of the way, we 
end up at another gnarly problem. The above procedure only succeeds if the seed is in $F$.
Quite naturally, one does not expect to get any cuts in $F$ by diffusing from a random
vertex in $G$. From the perspective of the global algorithm, this means that we need
some careful ordering of the seeds, so that low conductance cuts are discovered.
Unfortunately, we also need local implementations of this ordering. 
The authors struggled with carrying out this approach, but to no avail.

To rid ourselves of the ordering problem, let us consider the following, almost naive
global algorithm. First, order the vertices according to a uniform random permutation.
At any stage, there is a free set $F$. We process the next seed vertex $s$ by
running some spectral partitioning procedure, to get a low conductance cut $C(s)$.
Simply output $C(s) \cap F$ (instead of $C(s)$) as the new cluster, and update $F$ to $F \setminus C(s)$.
It is easy to locally implement this procedure.
To find the anchor of $v$, perform a diffusion of $\poly(\eps^{-1})$ timesteps from $v$.
For every vertex $s$ with high enough value in the diffusion vector, determine if $C(s) \ni v$.
The vertex $s$ that is lowest according to the random ordering is the anchor of $v$.
Unfortunately, there is little hope of bounding the number of edges cut
by the clustering. When $s$ is processed, it may be that $s \notin F$, and there is no guarantee
of $C(s) \cap F$. Can we modify the procedure to bound the number of cut edges, but
still maintain its ease of local implementability?

\medskip

{\bf The amortization argument:} Consider the scenario when $F = \Theta(\eps n)$.
Most of the subsequent seeds processed are not in $F$ and there is no guarantee
on the cluster conductance. But every $\Theta(1/\eps)$ seeds (in expectation),
we will get a ``good" seed $s$ contained in $F$, such that $C(s) \cap F$ is
a low conductance set. (This is promised by the diffusion algorithm that we develop in this paper,
as discussed earlier.) Our aim is to perform some amortization, to argue that
$|C(s) \cap F|$ is so large, that we can ``charge" away the edges cut
by the previous $\Theta(1/\eps)$ seeds.

This amortization is possible because our spectral tools give us much flexibility in
the (low) conductances obtained. Put differently, we essentially prove that existence
of many cuts of extremely low conductance, and show that it is ``easy" for a diffusion-based
algorithm to find such cuts. (This is connected to the spectral behavior
of minor-closed families.) As a consequence, we can actually pre-specify the size of the low
conductance cuts obtained. We show that as long as $|F| = \Omega(\eps n)$,
we can find a \emph{size threshold} $k = \poly(\eps^{-1})$ such that for at least $\Omega(\eps^2n)$ vertices $s \in F$,
a spectral partitioning procedure seeded at $s$ can find a cut 
of size $\Theta(k)$ and conductance at most $\eps^c$. Moreover, this cut is guaranteed to
contain at least $\eps^{c'} k$ vertices in $F$, despite the procedure being oblivious to $F$.
The parameter $c$ can be easily tuned, so we can increase $c$ arbitrarily while keeping $c'$ fixed, at the cost
of polynomial increases in running time. This tunability is crucial to our amortization argument.
We also show that given query access to $F$, a size threshold $k$
can be computed in $\poly(d\eps^{-1})$ time.

So when the global algorithm processes seed $s$, it runs the above spectral procedure
to try to obtain a set of size $\Theta(k)$ with conductance at most $\eps^c$. (If the 
procedure fails, the global algorithm simply set $C(s) = \{s\}$.) Thus,
we cut $O(\eps^c kd)$ edges for each seed processed. But after every $O(1/\eps)$ seeds,
we choose a ``good" seed such that $|C(s) \cap F| > \eps^{c'} k$. The total number
of edges cut is $O(\eps^c kd \times \eps^{-1}) = O(\eps^{c-1} kd)$. The total
number of new vertices clustered is at least $\eps^{c'}k$. Because we can tune parameters with
much flexibility, we can set $c \gg c'$. So the total number of edges
cut is $O(\eps^{c-c'-1}d)$ times the number of vertices clustered, where $c-c'-1 > 1$.
Overall, we will cut only $O(\eps nd)$ edges.

\medskip

{\bf Making it work through phases:} Unfortunately, as the process described above
continues, $F$ shrinks. Thus, the original choice of $k$ might not work, and
the guarantees on $|C(s) \cap F|$ for good seeds no longer hold. So we need
to periodically recompute the value of $k$. In a careful analysis, we show that this
recomputation is only required $\poly(\eps^{-1})$ times. Formally, we implement
the recomputation through \emph{phases}. Each vertex is independently assigned to one
of $\poly(\eps^{-1})$ phases. 
(Technically, we choose the phase of a vertex by sampling an independent geometric
random variable. We heavily use the memoryless property of the geometric distribution.)

For each phase, the value of $k$ is fixed. The local partition oracle will compute
these size thresholds for all phases, as a $\poly(d\eps^{-1})$ time preprocessing step.
The oracle (for $v$) runs a diffusion from $v$
to get a collection of candidate anchors. For each candidate $s$, the oracle determines
its phase, runs the spectral partitioning algorithm with correct phase parameters,
and determines if the candidate's low conductance cut contains $v$. 
The anchor is simply such a candidate of minimum phase, with ties broken by vertex id.

\subsection{Outline of sections} \label{sec:outline}

The algorithm description and proof has many moving parts, encapsulated by different sections.
\Sec{prelims} begins by discussing the truncated diffusion process, the main algorithmic tool
for partitioning. We then describe the global partitioning algorithm \globpart{}
(modulo a preprocessing step called \findr), which is far more convenient
to analyze. It will be readily apparent that this global procedure outputs a partition
of $G$ into connected components; the main challenge is to bound the number of edges cut.

Within \Sec{prelims}, we discuss how to implement \globpart{} by a local procedure. 
By ensuring that the output of the local procedure is identical to \globpart, we prove
the consistency property of \Def{oracle}. We then perform a fairly straightforward running time analysis,
which proves the running time property of \Def{oracle}.

The real heavy lifting begins in \Sec{findr}, where we describe the procedure \findr{} that 
computes the size thresholds. This section is devoted to proving salient properties of 
the size thresholds output by \findr. The analysis hinges on the diffusion and cut properties
stated in \Thm{restrict-cut}, which is the main tool connecting minor-freeness, diffusions,
and local partitioning.
\Sec{amort} uses all these tools to prove the cut bound of \globpart. At this stage, the complete
description and guarantees of the partition oracle are complete, modulo the proof of \Thm{restrict-cut}.

The proof of \Thm{restrict-cut} is split into sections. In \Sec{diffusion}, we use the hyperfiniteness
of minor-closed families to prove properties of truncated diffusions on minor-free families. 
\Sec{ls-cluster} has the key spectral
calculations, where the Lov\'{a}sz-Simonovits curve technique is used to find low conductance cuts.
This section has the crucial insights that allow for partitioning in the free set, using diffusions
in the overall graph.

\Sec{appl} has short proofs of the applications \Thm{testers} and \Thm{approx}. These are provided for
completeness, since identical calculations appear in the proof of Theorem 9.28 in~\cite{G17-book}.